\documentclass[12pt]{article}
\usepackage{amsxtra}
\usepackage{amssymb, amsmath}
\usepackage{amsthm}
\usepackage {hyperref}

\newtheorem{lemma}{Lemma}[section]
\newtheorem{theorem}{Theorem}[section]
\newtheorem{remark}{Remark}[section]

\usepackage[dvips]{graphicx}

\begin{document}
\begin{center}
\textbf{\LARGE{Sequences of Inequalities among Differences of Gini Means and Divergence Measures}}
\end{center}

\smallskip
\begin{center}
\textbf{\large{Inder Jeet Taneja}}\\
Departamento de Matem\'{a}tica\\
Universidade Federal de Santa Catarina\\
88.040-900 Florian\'{o}polis, SC, Brazil.\\
\textit{e-mail: taneja@mtm.ufsc.br\\
http://www.mtm.ufsc.br/$\sim $taneja}
\end{center}

\begin{abstract}
In 1938, Gini \cite{gin} studied a mean having two parameters. Later, many authors studied properties of this mean. In particular, it contains the famous means as harmonic, geometric, arithmetic, etc. Here we considered a sequence of inequalities arising due to particular values of each parameter of Gini's mean. This sequence generates many nonnegative differences. Not all of them are convex. We have studied here convexity of these differences and again established new sequences of inequalities of these differences. Considering in terms of probability distributions these differences, we have made connections with some of well known divergence measures.
\end{abstract}



\section{Gini Mean of order \textit{r} and \textit{s}}

The Gini \cite{gin} mean of order $r$ and $s$ is given by
\begin{equation}
\label{eq1}
E_{r,s} (a,b) =\begin{cases}
 {\left( {\frac{a^r + b^r}{a^s + b^s}} \right)^{\frac{1}{r - s}}} & {r \ne
s} \\
 {\exp \left( {\frac{a^r\ln a + b^r\ln b}{a^r + b^r}} \right)} & {r = s \ne
0} \\
 {\sqrt {ab} } & {r = s = 0} \\
\end{cases}.
\end{equation}

In particular when $s = 0$ in (\ref{eq1}), we have
\begin{equation}
\label{eq2}
E_{r,0} (a,b): = B_r (a,b) =\begin{cases}
 {\left( {\frac{a^r + b^r}{2}} \right)^{\frac{1}{r}},} & {r \ne 0} \\
 {\sqrt {ab} ,} & {r = 0} \\
\end{cases}
\end{equation}

Again, when $s = r - 1$ in (\ref{eq1}), we have
\begin{equation}
\label{eq3}
E_{r,r - 1} (a,b): = K_s (a,b) = \frac{a^r + b^r}{a^{r - 1} + b^{r - 1}},\;r
\in {\rm R}
\end{equation}

The expression (\ref{eq2}) is famous \textit{as mean of order r} or \textit{power mean}. The expression (\ref{eq3}) is known as \textit{Lehmer mean}\cite{leh}. Both these means are monotonically increasing in $r$. Moreover, these two have the following inequality \cite{che} among each other:
\begin{equation}
\label{eq4}
B_r (a,b)\begin{cases}
 { < L_r (a,b),} & {r > 1} \\
 { > L_r (a,b),} & {r < 1} \\
\end{cases}
\end{equation}

Since $E_{r,s} = E_{s,r} $, the Gini-mean $E_{r,s} (a,b)$ given by (\ref{eq1}) is an increasing function in $r$ or $s$. Using the monotonicity property \cite{czp}, \cite{sim}, \cite{san} we have the following inequalities:
\begin{align}
\label{eq5}
E_{ - 3, - 2}  \le &  E_{ - 2, - 1} \le E_{ - 3 / 2, - 1 / 2} \le E_{ - 1,0} \le E_{ - 1 / 2,0} \notag\\
& \le E_{ - 1 / 2,1 / 2} \le E_{0,1 / 2} \le E_{0,1} \le E_{0,2} \le E_{1,2}
\end{align}

\noindent or
\begin{align}
\label{eq6}
E_{ - 3, - 2} \le & E_{ - 2, - 1} \le E_{ - 3 / 2, - 1 / 2} \le E_{ - 1,0} \le E_{ - 1 / 2,0} \notag\\
&\le E_{ - 1 / 2,1 / 2} \le E_{0,1 / 2} \le E_{0,1} \le E_{1 /2,1} \le E_{1,2}
\end{align}

Equivalently, let us write as
\begin{align}
\label{eq7}
K_{ - 2} \le & K_{ - 1} \le  K_{ - 1 / 2} \le \left( {K_0 = B_{ - 1} = H}
\right) \le B_{ - 1 / 2} \notag\\
& \le \left( {K_{1 / 2} = B_0 = G} \right)  \le B_{1 / 2} \le \left( {K_1 = B_1 = A} \right) \le \left( {B_2 = S} \right) \le K_2 ,
\end{align}

\noindent or
\begin{align}
\label{eq8}
K_{ - 2} \le & K_{ - 1} \le  K_{ - 1 / 2} \le \left( {K_0 = B_{ - 1} = H}
\right) \le B_{ - 1 / 2} \notag\\
& \le \left( {K_{1 / 2} = B_0 = G} \right) \le B_{1 /2}
\le \left( {K_1 = B_1 = A} \right) \le E_{1 / 2,1} \le K_2 .
\end{align}

\noindent
where $H$, $G$, $A$ and $S$ are respectively , the \textit{harmonic}, \textit{geometric}, \textit{arithmetic} and the \textit{square-root} means. We observe from the expression (\ref{eq7}) that the means considered either are the particular cases of (\ref{eq2}) or of (\ref{eq3}), while in the second case we have we have a mean $E_{1 / 2,1} $, which is neither a particular case of (\ref{eq2}) nor of (\ref{eq3}). We can easily find values proving that there is no relation between $S$ and $E_{1 / 2,1} $. Summarizing the expressions (\ref{eq7}) and (\ref{eq8}), we can write them in joint form as
\begin{equation}
\label{eq9}
P_1 \le P_2 \le P_3 \le H \le P_4 \le G \le N_1 \le A \le \left( {P_5 \mbox{
or }S} \right) \le P_6 ,
\end{equation}

\noindent
where $P_1 = K_{ - 2} $, $P_2 = K_{ - 1} $, $P_3 = K_{ - 1 / 2} $, $P_4 =
B_{ - 1 / 2} $, $N_1 = B_{1 / 2} $, $P_5 = E_{1 / 2,1} $ and $P_5 = K_2 $.

\bigskip
In \cite{tan2, tan3}, the author studied the following inequalities:
\begin{equation}
\label{eq10}
H \le G \le N_1 \le N_3 \le N_2 \le A \le S,
\end{equation}

\noindent where
\[
N_2 (a,b) = \left( {\frac{\sqrt a + \sqrt b }{2}} \right)\left( {\sqrt
{\frac{a + b}{2}} } \right)
\]

\noindent and
\[
N_3 (a,b) = \frac{a + \sqrt {ab} + b}{3}.
\]

The expression $N_3 (a,b)$ is famous as Heron's mean. The inequalities
(\ref{eq10}) admit many non-negative differences. Bases on these difference the
author [\ldots ] proved the following result:
\begin{equation}
\label{eq11}
D_{SA} \le \left\{ {\begin{array}{l}
 \textstyle{1 \over 3}D_{SH} \le \left\{ {\begin{array}{l}
 \textstyle{1 \over 2}D_{AH} \le \left\{ {\begin{array}{l}
 4D_{N_2 N_1 } \le \textstyle{4 \over 3}D_{N_2 G} \le D_{AG} \le 4D_{AN_2}
\\
 \textstyle{1 \over 2}D_{SG} \le D_{AG} \\
 \end{array}} \right. \\\\
 \textstyle{2 \over 3}D_{SN_1 } \le \textstyle{1 \over 2}D_{SG} \\
 \end{array}} \right. \\
 \\
 \textstyle{4 \over 5}D_{SN_2 } \le 4D_{AN_2 } \\
 \textstyle{3 \over 4}D_{SN_3 } \le \textstyle{2 \over 3}D_{SN_1 } \\
 \end{array}} \right.,
\end{equation}

\noindent
where $D_{SA} = S - A$, $D_{SH} = S - H$, etc. Some applications of the
inequalities (\ref{eq11}) can be seen in \cite{szd}, \cite{sim}.

\bigskip
Combining (\ref{eq9}) and (\ref{eq10}), we have the following sequence of inequalities:
\begin{equation}
\label{eq12}
P_1 \leqslant P_2 \leqslant P_3 \leqslant H \leqslant P_4 \leqslant G
\leqslant N_1 \leqslant N_3 \leqslant N_2 \leqslant A \leqslant \left( {P_5
\mbox{ or }S} \right) \leqslant P_6 .
\end{equation}

The inequalities (\ref{eq12}) admits many nonnegative differences such as
\begin{equation}
\label{eq13}
D_{tp} (a,b) = bg_{tp} \left( {\frac{a}{b}} \right) = b\left[ {f_t \left(
{\frac{a}{b}} \right) - f_p \left( {\frac{a}{b}} \right)} \right],
\end{equation}

\noindent where
\[
g_{tp} (x) = f_t (x) - f_p (x),
\,
f_t (x) \geqslant f_p (x),
\,
\forall x > 0.
\]

More precisely, the expression (\ref{eq12}) and the function $f:(0,\infty ) \to
\mathbb{R}$ given in (\ref{eq13}) lead us to the following result:
\[
f_{P_1 } (x) \leqslant f_{P_2 } (x) \leqslant f_{P_3 } (x) \leqslant f_H (x)
\leqslant f_{P_4 } (x) \leqslant f_G (x) \leqslant f_{N_1 } (x)
\]
\begin{equation}
\label{eq14}
 \leqslant f_{N_3 } (x) \leqslant f_{N_2 } (x) \leqslant f_A (x) \leqslant
\left( {f_{P_5 } (x)\mbox{ or }f_S (x)} \right) \leqslant f_{P_6 } (x).
\end{equation}

Equivalently,
\begin{align}
& \frac{x(x^2 + 1)}{x^3 + 1} \leqslant \frac{x(x + 1)}{x^2 + 1} \leqslant
\frac{x\left( {\sqrt x + 1} \right)}{x^{3 / 2} + 1} \leqslant \frac{2x}{1 +
x} \leqslant \frac{4x}{\left( {\sqrt x + 1} \right)^2} \leqslant\notag\\
& \hspace{15pt} \leqslant \sqrt x \leqslant \left( {\frac{\sqrt x + 1}{2}} \right)^2
\leqslant \frac{x + \sqrt x + 1}{3} \leqslant \left( {\frac{\sqrt x + 1}{2}}
\right)\left( {\sqrt {\frac{x + 1}{2}} } \right) \leqslant \notag\\
\label{eq15}
& \hspace{30pt} \leqslant \frac{x + 1}{2} \leqslant \left( {\left( {\frac{x + 1}{\sqrt x +
1}} \right)^2\mbox{ or }\sqrt {\frac{x^2 + 1}{2}} } \right) \leqslant
\frac{x^2 + 1}{x + 1}.
\end{align}

In this paper our aim is to produce new inequalities for the difference of
means arising due to inequalities given in (\ref{eq12}). In another words we shall
improve considerably the inequalities given in (\ref{eq11}). For this we need
first to know the convexity of the difference of means. In total, we have 77
differences. Some of them are equal to each other with some multiplicative
constants. Some of them are not convex and some of them are convex.

\section{Convexity of Difference of Means}

Let us prove now the convexity of some of the difference of means arising
due to inequalities (\ref{eq11}). In order to prove it we shall make use of the
following lemma (ref. Taneja [\ldots ]).

\begin{lemma} Let $f:I \subset \mathbb{R}_ + \to \mathbb{R}$ be a
convex and differentiable function satisfying $f(1) = 0$. Consider a
function
\[
\phi _f (a,b) = af\left( {\frac{b}{a}} \right),
\, a,b > 0,
\]

\noindent
then the function $\phi _f (a,b)$ is convex in $\mathbb{R}_ + ^2 $.
Additionally, if $f^\prime (1) = 0$, then the following inequality hold:
\[
0 \leqslant \phi _f (a,b) \leqslant \left( {\frac{b - a}{a}} \right)\phi
_{f}' (a,b).
\]
\end{lemma}

Inequalities appearing in (\ref{eq12}) admits 77 nonnegative differences. Within
these we have the following two set of equalities with multiplicative
constants:

\bigskip
\noindent (i) $D_{P_6 A} = D_{P_6 H} = \frac{1}{2}D_{AH} $.\\

\noindent (ii) $D_{P_5 P_4 } = D_{AN_3 } = \frac{2}{3}D_{AN_1 } = \frac{1}{3}D_{AG} =
2D_{N_3 N_1 } = \frac{1}{2}D_{N_3 G} = \frac{2}{3}D_{N_1 G} $.

\bigskip
Not all the differences of means appearing in (\ref{eq12}) are convex. We shall
consider only those are convex. It is easy to check that in all the cases,
$f_{( \cdot )} (1) = 1$, i.e., $g_{tp} (1) = 0$, $t > p$. According to Lemma
2.1, it is suficient to show the convexity of the functions $g_{tp} (x)$,
i.e., to show that the second order derivative of $g_{tp} (x)$ is
nonnegative for all $x > 0$.

\begin{enumerate}
\item \textbf{For }$\bf{D_{P_6 S} (a,b)}$\textbf{:} We can write$D_{P_6 S} (a,b) =
b\,g_{P_6 S} \left( {a \mathord{\left/ {\vphantom {a b}} \right.
\kern-\nulldelimiterspace} b} \right)$, where
\[
g_{P_6 S} (x) = f_{P_6 } (x) - f_S (x) = \frac{x^2 + 1}{x + 1} - \sqrt
{\frac{x^2 + 1}{2}} .
\]

This gives
\begin{equation}
\label{eq16}
{g}''_{P_6 S} (x) = \frac{2\left[ {2\left( {2x^2 + 2} \right)^{3 / 2} -
\left( {x + 1} \right)^3} \right]}{(x + 1)^3\left( {2x^2 + 2} \right)^{3 /
2}}.
\end{equation}

Since, $S \geqslant A$, this implies that $S^3 \geqslant A^3$, i.e., $\left(
{\sqrt {\frac{x^2 + 1}{2}} } \right)^3 - \left( {\frac{x + 1}{2}} \right)^3
\geqslant 0$. This gives $2\left( {2x^2 + 2} \right)^{3 / 2} - \left( {x +
1} \right)^3 \geqslant 0$. Thus we have ${g}''_{P_6 S} (x) \geqslant 0$ for
all $x > 0$.

\item \textbf{For }$\bf{D_{P_6 N_2 } (a,b)}$\textbf{:} We can write $D_{P_6 N_2 } (a,b)
= b\,g_{P_6 N_2 } \left( {a \mathord{\left/ {\vphantom {a b}} \right.
\kern-\nulldelimiterspace} b} \right)$, where
\[
g_{P_6 N_2 } (x) = f_{P_6 } (x) - f_{N_2 } (x)
 = \frac{4(x^2 + 1) - \sqrt {2x + 2} \left( {\sqrt x + 1} \right)\left( {x +
1} \right)}{4(x + 1)}.
\]

This gives
\begin{equation}
\label{eq17}
{g}''_{P_6 N_2 } (x) = \frac{\left( {x + 1} \right)^3\left( {x^{3 / 2} + 1}
\right) + 16x^{3 / 2}\left( {2x + 2} \right)^{3 / 2}}{4x^{3 / 2}(x + 1)^3(2x
+ 2)^{3 / 2}} > 0,
\,
\forall x > 0.
\end{equation}

\item \textbf{For }$\bf{D_{P_6 N_3 } (a,b)}$\textbf{:} We can write $D_{P_6 N_3 } (a,b)
= b\,g_{P_6 N_3 } \left( {a \mathord{\left/ {\vphantom {a b}} \right.
\kern-\nulldelimiterspace} b} \right)$, where
\[
g_{P_6 N_3 } (x) = f_{P_6 } (x) - f_{N_3 } (x)
 = \frac{2(x^2 + 1) - \sqrt x \left( {\sqrt x + 1} \right)^2}{3(x + 1)}.
\]

This gives
\begin{equation}
\label{eq18}
{g}''_{P_6 N_3 } (x) = \frac{48x^{3 / 2} + \left( {x + 1} \right)^3}{12x^{3
/ 2}(x + 1)^3} > 0,
\,
\forall x > 0.
\end{equation}

\item \textbf{For }$\bf{D_{P_6 N_1 } (a,b)}$\textbf{:} We can write $D_{P_6 N_1 } (a,b)
= b\,g_{P_6 N_1 } \left( {a \mathord{\left/ {\vphantom {a b}} \right.
\kern-\nulldelimiterspace} b} \right)$, where
\[
g_{P_6 N_1 } (x) = f_{P_6 } (x) - f_{N_1 } (x)
 = \frac{3(x^2 + 1) - 2\sqrt x \left( {x + \sqrt x + 1} \right)}{4(x + 1)}.
\]

This gives
\begin{equation}
\label{eq19}
{g}''_{P_6 N_1 } (x) = \frac{32x^{3 / 2} + (x + 1)^3}{8x^{3 / 2}(x + 1)^3} >
0,
\,
\forall x > 0.
\end{equation}

\item \textbf{For }$\bf{D_{P_6 G} (a,b)}$\textbf{:} We can write $D_{P_6 G} (a,b) =
b\,g_{P_6 G} \left( {a \mathord{\left/ {\vphantom {a b}} \right.
\kern-\nulldelimiterspace} b} \right)$, where
\begin{align}
g_{P_6 G} (x) & = f_{P_6 } (x) - f_G (x)
 = \frac{x^2 + 1}{x + 1} - \sqrt x \notag\\
 &= \frac{\left( {\sqrt x - 1}
\right)^2\left( {x + \sqrt x + 1} \right)}{x + 1}.\notag
\end{align}

This gives
\begin{equation}
\label{eq20}
{g}''_{P_6 G} (x) = \frac{16x^{3 / 2} + (x + 1)^3}{4x^{3 / 2}(x + 1)^3} > 0,
\,
\forall x > 0.
\end{equation}

\item \textbf{For }$\bf{D_{P_6 P_4 } (a,b)}$\textbf{:} We can write$D_{P_6 P_4 } (a,b)
= b\,g_{P_6 P_4 } \left( {a \mathord{\left/ {\vphantom {a b}} \right.
\kern-\nulldelimiterspace} b} \right)$, where
\begin{align}
g_{P_6 P_4 } (x) & = f_{P_6 } (x) - f_{P_4 } (x) = \frac{x^2 + 1}{x + 1} -
\frac{4x}{\left( {\sqrt x + 1} \right)^2}\notag\\
& = \frac{\left( {x + 1} \right)\left( {x^2 - 4x + 1} \right) + 2\sqrt x
\left( {x^2 + 1} \right)}{\left( {x + 1} \right)\left( {\sqrt x + 1}
\right)^2}.\notag
\end{align}

This gives
\begin{equation}
\label{eq21}
{g}''_{P_6 P_4 } (x) = \frac{2\left[ {3\left( {x^3 + 1} \right) + 17x(x + 1)
+ 2\sqrt x \left( {x^2 + 6x + 1} \right)} \right]}{\sqrt x \left( {\sqrt x +
1} \right)^4\left( {x + 1} \right)^3} > 0,
\,
\forall x > 0.
\end{equation}

\item \textbf{For }$\bf{D_{P_6 P_2 } (a,b)}$\textbf{:} We can write$D_{P_6 P_2 } (a,b)
= b\,g_{P_6 P_2 } \left( {a \mathord{\left/ {\vphantom {a b}} \right.
\kern-\nulldelimiterspace} b} \right)$, where
\[
g_{P_6 P_2 } (x) = f_{P_6 } (x) - f_{P_2 } (x)
 = \frac{\left( {x - 1} \right)^2\left( {x^2 + x + 1} \right)}{\left( {x +
1} \right)\left( {x^2 + 1} \right)}.
\]

This gives
\begin{equation}
\label{eq22}
{g}''_{P_6 P_2 } (x) = \frac{2\left( {x^6 + 15x^4 + 16x^3 + 15x^2 + 1}
\right)}{\left( {x + 1} \right)^3\left( {x^2 + 1} \right)^3} > 0,
\,
\forall x > 0.
\end{equation}

\item \textbf{For }$\bf{D_{P_6 P_1 } (a,b)}$\textbf{:} We can write $D_{P_6 P_1 } (a,b)
= b\,g_{P_6 P_2 } \left( {a \mathord{\left/ {\vphantom {a b}} \right.
\kern-\nulldelimiterspace} b} \right)$, where
\[
g_{P_6 P_1 } (x) = f_{P_6 } (x) - f_{P_1 } (x)
 = \frac{\left( {x - 1} \right)^2\left( {x^2 + 1} \right)}{\left( {x + 1}
\right)^3}.
\]

This gives
\begin{equation}
\label{eq23}
{g}''_{P_6 P_1 } (x) = \frac{2\left( {x^2 + 1} \right)\left(
{\begin{array}{l}
 \left( {\left( {x - 1} \right)^2 - x} \right)^2 + x^2 + \\
 + \left( {x - 1} \right)^4 + x\left( {x - 1} \right)^2 \\
 \end{array}} \right) + 8x^3}{(x^3 + 1)^3} > 0,\,\forall x > 0.
\end{equation}

\item \textbf{For }$\bf{D_{P_5 A} (a,b)}$\textbf{:} We can write $D_{P_5 A} (a,b) =
b\,g_{P_5 A} \left( {a \mathord{\left/ {\vphantom {a b}} \right.
\kern-\nulldelimiterspace} b} \right)$, where
\begin{align}
g_{P_5 A} (x) & = f_{P_5 } (x) - f_A (x)
 = \left( {\frac{x + 1}{\sqrt x + 1}} \right)^2 - \frac{x + 1}{2} \notag\\
 & = \frac{\left( {x + 1} \right)\left( {\sqrt x - 1} \right)^2}{2\left( {\sqrt x
+ 1} \right)^2}.\notag\
\end{align}

This gives
\begin{equation}
\label{eq24}
{g}''_{P_5 A} (x) = \frac{4\sqrt x \left[ {\left( {\sqrt x - 1} \right)^2 +
\sqrt x } \right] + \left( {x - 1} \right)^2}{2x^{3 / 2}\left( {\sqrt x + 1}
\right)^4} > 0,\, \forall x > 0.
\end{equation}

\item \textbf{For }$\bf{D_{P_5 N_2 } (a,b)}$\textbf{:} We can write $D_{P_5 N_2 } (a,b)
= b\,g_{P_5 N_2 } \left( {a \mathord{\left/ {\vphantom {a b}} \right.
\kern-\nulldelimiterspace} b} \right)$, where
\[
g_{P_5 N_2 } (x) = f_{P_5 } (x) - f_{N_2 } (x)
 = \left( {\frac{x + 1}{\sqrt x + 1}} \right)^2 - \left( {\frac{\sqrt x +
1}{2}} \right)\left( {\sqrt {\frac{x + 1}{2}} } \right).
\]

This gives
\begin{align}
{g}''_{P_5 N_2 } (x) & = \frac{1}{8\left( {\sqrt x + 1} \right)x^{3 / 2}\left(
{x + 1} \right)\sqrt {2x + 2} }\times \notag\\
\label{eq25}
& \hspace{15pt} \times \left( {\begin{array}{l}
 \left( {\sqrt x + 1} \right)^4\left( {x^{3 / 2} + 1} \right) + \,4\sqrt {2x
+ 2} \left( {x + 1} \right)\times \\
 \times \left[ {4\sqrt x \left( {\left( {\sqrt x - 1} \right)^2 + \sqrt x }
\right) + \left( {x - 1} \right)^2} \right] \\
 \end{array}} \right) > 0,
\,
\forall x > 0.
\end{align}

\item \textbf{For }$\bf{D_{P_5 N_3 } (a,b)}$\textbf{:} We can write $D_{P_5 N_3 } (a,b)
= b\,g_{P_5 N_3 } \left( {a \mathord{\left/ {\vphantom {a b}} \right.
\kern-\nulldelimiterspace} b} \right)$, where
\begin{align}
g_{P_5 N_3 } (x) & = f_{P_5 } (x) - f_{N_3 } (x)
 = \left( {\frac{x + 1}{\sqrt x + 1}} \right)^2 - \frac{x + \sqrt x + 1}{3}\notag\\
& = \frac{\left( {\sqrt x - 1} \right)^2\left( {2x + \sqrt x + 2}
\right)}{2\left( {\sqrt x + 1} \right)^2}.\notag
\end{align}

This gives
\begin{equation}
\label{eq26}
{g}''_{P_5 N_3 } (x) = \frac{7\left( {\sqrt x - 1} \right)^2\left( {x +
6\sqrt x + 1} \right) + 40x}{12x^{3 / 2}\left( {\sqrt x + 1} \right)^4} > 0,
\,
\forall x > 0.
\end{equation}

\item \textbf{For }$\bf{D_{P_5 N_1 } (a,b)}$\textbf{:} We can write$D_{P_5 N_1 } (a,b)
= b\,g_{P_5 N_1 } \left( {a \mathord{\left/ {\vphantom {a b}} \right.
\kern-\nulldelimiterspace} b} \right)$, where
\begin{align}
g_{P_5 N_1 } (x)   &= f_{P_5 } (x) - f_{N_1 } (x)
 = \left( {\frac{x + 1}{\sqrt x + 1}} \right)^2 - \left( {\frac{\sqrt x +
1}{2}} \right)^2\notag\\
& = \frac{\left( {\sqrt x - 1} \right)^2\left( {3x + 2\sqrt x + 3}
\right)}{4\left( {\sqrt x + 1} \right)^2}.\notag
\end{align}

This gives
\begin{equation}
\label{eq27}
{g}''_{P_5 N_1 } (x) = \frac{5\left( {\sqrt x - 1} \right)^2\left( {x +
6\sqrt x + 1} \right) + 32x}{8x^{3 / 2}\left( {\sqrt x + 1} \right)^4} > 0,
\,
\forall x > 0.
\end{equation}

\item \textbf{For }$\bf{D_{P_5 G} (a,b)}$\textbf{:} We can write $D_{P_5 G} (a,b) =
b\,g_{P_5 G} \left( {a \mathord{\left/ {\vphantom {a b}} \right.
\kern-\nulldelimiterspace} b} \right)$, where
\begin{align}
g_{P_5 G} (x) & = f_{P_5 } (x) - f_G (x) = \left( {\frac{x + 1}{\sqrt x + 1}}
\right)^2 - \sqrt x\notag\\
& = \frac{\left( {\sqrt x - 1} \right)^2\left( {x + \sqrt x + 1}
\right)}{\left( {\sqrt x + 1} \right)^2}.\notag
\end{align}

This gives
\begin{equation}
\label{eq28}
{g}''_{P_5 G} (x) = \frac{3\left[ {4\sqrt x \left( {x + 1} \right) + \left(
{x - 1} \right)^2} \right]}{4x^{3 / 2}\left( {\sqrt x + 1} \right)^4} > 0,
\,
\forall x > 0.
\end{equation}

\item \textbf{For }$\bf{D_{P_5 H} (a,b)}$\textbf{:} We can write $D_{P_5 H} (a,b) =
b\,g_{P_5 H} \left( {a \mathord{\left/ {\vphantom {a b}} \right.
\kern-\nulldelimiterspace} b} \right)$, where
\[
g_{P_5 H} (x) = f_{P_5 } (x) - f_H (x) = \left( {\frac{x + 1}{\sqrt x + 1}}
\right)^2 - \frac{2x}{x + 1}.
\]

This gives
\begin{align}
{g}''_{P_5 H} (x) & = \frac{1}{2x^{3 / 2}\left( {x + 1} \right)^3\left( {\sqrt
x + 1} \right)^4}\times\notag\\
\label{eq29}
& \hspace{15pt} \times \left( {\begin{array}{l}
 \left( {x + 1} \right)\left[ {\left( {x - 1} \right)^4 + 16x^2} \right] +
\\
 + \,4\sqrt x \left[ {\left( {x + 1} \right)^4 + 2x\left( {x^2 + 6x + 1}
\right)} \right] \\
 \end{array}} \right) > 0,
\,
\forall x > 0.
\end{align}

\item \textbf{For }$\bf{D_{P_5 P_3 } (a,b)}$\textbf{:} We can write $D_{P_5 P_3 } (a,b)
= b\,g_{P_5 P_3 } \left( {\frac{a}{b}} \right)$, where
\begin{align}
g_{P_5 P_3 } & = f_{P_5 } (x) - f_{P_3 } (x)
= \left( {\frac{x + 1}{\sqrt x + 1}} \right)^2 - \frac{x\left( {\sqrt x +
1} \right)}{x^{3 / 2} + 1}\notag\\
& = \frac{\left( {\sqrt x - 1} \right)^2\left( {x^2 + x^{3 / 2} + 3x + \sqrt
x + 1} \right)}{\left( {\sqrt x + 1} \right)\left( {x^{3 / 2} + 1}
\right)}. \notag
\end{align}

This gives
\begin{equation}
\label{eq30}
{g}''_{P_5 P_3 } (x) = \frac{s_1 (x)}{4x^{3 / 2}\left( {x^{3 / 2} + 1}
\right)^3\left( {\sqrt x + 1} \right)},
\end{equation}

\noindent
where
\[
s_1 (x) = \left( {\begin{array}{l}
 2x^5 + 2x^{9 / 2} - 27x^4 + 75x^{7 / 2} - 123x^3 + \\
 + 198x^{5 / 2} - 123x^2 + 75x^{3 / 2} - 27x + 2\sqrt x + 2 \\
 \end{array}} \right).
\]

Now we shall show that $s_1 (x) > 0$, $\forall x > 0$. Let us consider
\[
h_1 (t) = s_1 (t^2) = \left( {\begin{array}{l}
 2t^{10} + 2t^9 - 27t^8 + 75t^7 - 123t^6 + \\
 + 198t^5 - 123t^4 + 75t^3 - 27t^2 + 2t + 2 \\
 \end{array}} \right).
\]

The polynomial equation $h(t) = 0$ of 10$^{th}$ degree admits 10 solutions.
Out of them 8 are complex and are given by
\begin{align}
& 0.002340265888\pm 0.6477415967\;I;\,\,0.005577710358\pm 1.543805357\;I;\notag\\
& 0.4030868619\pm 0.2323740013\;I;\,\quad 1.862033520\pm 1.073436573\;I.\notag
\end{align}

The real solutions are $ - 5.359491864$ and $ - 0.1865848527$. Both these
solutions are negative. Since we are working with $t > 0$, this means that
there are no real positive solutions of the equation $h_1 (t) = 0$. Thus we
conclude that either $h_1 (t) > 0$ or $h_1 (t) < 0$, for all $t > 0$. In
order to check it is sufficient to see for any particular value of $h_1
(t)$, for example when $t = 1$. This gives $h_1 (1) = 56$, hereby proving
that $h_1 (t) > 0$ for all $t > 0$, consequently, $s_1 (x) > 0$, for all $x
> 0$. Finally, we have ${g}''_{P_5 P_3 } (x) > 0$, $\forall x > 0$.

\item \textbf{For }$\bf{D_{P_5 P_2 } (a,b)}$\textbf{:} We can write $D_{P_5 P_2 } (a,b)
= b\,g_{P_5 P_2 } \left( {\frac{a}{b}} \right)$, where
\begin{align}
g_{P_5 P_2 } & = f_{P_5 } (x) - f_{P_2 } (x)
 = \left( {\frac{x + 1}{\sqrt x + 1}} \right)^2 - \frac{x\left( {x + 1}
\right)}{x^2 + 1}\notag\\
&  = \frac{\left( {x + 1} \right)\left( {x^{3 / 2} - 1} \right)^2}{\left(
{\sqrt x + 1} \right)^2\left( {x^2 + 1} \right)}.\notag
\end{align}

This gives
\begin{equation}
\label{eq31}
{g}''_{P_5 P_2 } (x) = \frac{s_2 (x)}{2x^{3 / 2}\left( {x^2 + 1}
\right)^3\left( {\sqrt x + 1} \right)^4},
\end{equation}

\noindent
where
\[
s_2 (x) = \left( {\begin{array}{l}
 x^8 + 4x^{15 / 2} - 6x^7 - 12x^6 + 14x^5 + 92x^{9 / 2} + \\
 + 102x^4 + 92x^{7 / 2} + 14x^3 - 12x^2 - 6x + 4\sqrt x + 1 \\
 \end{array}} \right).
\]

Now we shall show that $s_2 (x) > 0$, $\forall x > 0$. Let us consider
\[
h_2 (t) = s_2 (t^2) = \left( {\begin{array}{l}
 t^{16} + 4t^{15} - 6t^{14} - 12t^{12} + 14t^{10} + 92t^9 + \\
 + 102t^8 + 92t^7 + 14t^6 - 12t^4 - 6t^2 + 4t + 1 \\
 \end{array}} \right).
\]

The polynomial equation $h_2 (t) = 0$ of 16$^{th}$ degree admits 16
solutions. Out of them 14 are complex (not written here) and two of them are
real given by $ - 5.230974171$ and $ - 0.1911689806$. Both these solutions
are negative. Since we are working with $t > 0$, this means that there are
no real positive solutions of the equation $h_2 (t) = 0$. Thus we conclude
that either $h_2 (t) > 0$ or $h_2 (t) < 0$, for all $t > 0$. In order to
check it is sufficient to see for any particular value of $h_2 (t)$, for
example when $t = 1$. This gives $h_2 (1) = 288$, hereby proving that $h_2
(t) > 0$ for all $t > 0$, consequently, $s_2 (x) > 0$, for all $x > 0$.
Finally, we have ${g}''_{P_5 P_2 } (x) > 0$, $\forall x > 0$.

\item \textbf{For }$\bf{D_{P_5 P_1 } (a,b)}$\textbf{:} We can write $D_{P_5 P_1 } (a,b)
= b\,g_{P_5 P_1 } \left( {\frac{a}{b}} \right)$, where
\begin{align}
g_{P_5 P_1 } & = f_{P_5 } (x) - f_{P_1 } (x)
 = \left( {\frac{x + 1}{\sqrt x + 1}} \right)^2 - \frac{x\left( {x^2 + 1}
\right)}{x^3 + 1}\notag\\
& = \frac{\left( {\begin{array}{l}
 x^4 + 2x^{7 / 2} + 4x^3 + 4x^{5 / 2} + \\
 + 4x^2 + 4x^{3 / 2} + 4x + 2\sqrt x + 1 \\
 \end{array}} \right)\left( {\sqrt x - 1} \right)^2}{\left( {\sqrt x + 1}
\right)^2\left( {x^{3 / 2} + 1} \right)}.\notag
\end{align}

This gives
\begin{equation}
\label{eq32}
{g}''_{P_5 P_1 } (x) = \frac{s_3 (x)}{2x^{3 / 2}\left( {x^3 + 1}
\right)^3\left( {\sqrt x + 1} \right)^4},
\end{equation}

\noindent
where
\[
s_3 (x) = \left( {\begin{array}{l}
 x^{11} + 4x^{21 / 2} - 6x^{10} + 4x^{19 / 2} + x^9 - 12x^{17 / 2} - \\
 - 45x^8 - 36x^{15 / 2} + 30x^7 + 144x^{13 / 2} + 99x^6 + \\
 + 48x^{11 / 2} + 99x^5 + 144x^{9 / 2} + 30x^4 - 36x^{7 / 2} - \\
 - 45x^3 - 12x^{5 / 2} + x^2 + 4x^{3 / 2} - 6x + 4\sqrt x + 1 \\
 \end{array}} \right).
\]

Now we shall show that $s_3 (x) > 0$, $\forall x > 0$. Let us consider
\[
h_3 (t) = s_3 (t^2) = \left( {\begin{array}{l}
 t^{22} + 4t^{21} - 6t^{20} + 4t^{19} + t^{18} - 12t^{17} - \\
 - 45t^{16} - 36t^{15} + 30t^{14} + 144t^{13} + 99t^{12} + \\
 + 48t^{11} + 99t^{10} + 144t^9 + 30t^8 - 36t^7 - \\
 - 45t^6 - 12t^5 + t^4 + 4t^3 - 6t^2 + 4t + 1 \\
 \end{array}} \right).
\]

The polynomial equation $h_3 (t) = 0$ of 22$^{nd}$ degree admits 22
solutions. Out of them 20 are complex (not written here) and two of them are
real given by $ - 5.269597986$ and $ - 0.1897677968$. Both these solutions
are negative. Since we are working with $t > 0$, this means that there are
no real positive solutions of the equation $h_3 (t) = 0$. Thus we conclude
that either $h_3 (t) > 0$ or $h_3 (t) < 0$, for all $t > 0$. In order to
check it is sufficient to see for any particular value of $h_3 (t)$, for
example when $t = 1$. This gives $h_3 (1) = 416$, hereby proving that $h_3
(t) > 0$ for all $t > 0$, consequently, $s_3 (x) > 0$, for all $x > 0$.
Finally, we have ${g}''_{P_5 P_1 } (x) > 0$, $\forall x > 0$.

\item \textbf{For }$\bf{D_{SP_4 } (a,b)}$\textbf{:} We can write $D_{SP_4 } (a,b) =
b\,g_{SP_4 } \left( {\frac{a}{b}} \right)$, where
\[
g_{SP_4 } = f_S (x) - f_{P_4 } (x)
 = \sqrt {\frac{x^2 + 1}{2}} - \frac{4x}{\left( {\sqrt x + 1} \right)^2}.
\]

This gives
\begin{equation}
\label{eq33}
{g}''_{SP_4 } (x) = \frac{2\left[ {\sqrt x \left( {\sqrt x + 1} \right)^4 +
3\left( {2x^2 + 2} \right)^{3 / 2}} \right]}{\sqrt x \left( {\sqrt x + 1}
\right)^4\left( {2x^2 + 2} \right)^{3 / 2}} > 0,
\,
\forall x > 0.
\end{equation}

\item \textbf{For }$\bf{D_{AP_4 } (a,b)}$\textbf{:} We can write $D_{AP_4 } (a,b) =
b\,g_{AP_4 } \left( {a \mathord{\left/ {\vphantom {a b}} \right.
\kern-\nulldelimiterspace} b} \right)$, where
\[
g_{AP_4 } (x) = f_A (x) - f_{P_4 } (x)
 = \frac{x + 1}{2} - \frac{4x}{\left( {\sqrt x + 1} \right)^2}.
\]

This gives
\begin{equation}
\label{eq34}
{g}''_{AP_4 } (x) = \frac{6}{\sqrt x \left( {\sqrt x + 1} \right)^4} > 0,
\,
\forall x > 0.
\end{equation}

\item \textbf{For }$\bf{D_{SA} (a,b)}$\textbf{:} We can write $D_{SA} (a,b) = b\,g_{SA}
\left( {a \mathord{\left/ {\vphantom {a b}} \right.
\kern-\nulldelimiterspace} b} \right)$, where
\[
g_{SA} (x) = f_S (x) - f_A (x)
 = \sqrt {\frac{x^2 + 1}{2}} - \frac{x + 1}{2}.
\]

This gives
\begin{equation}
\label{eq35}
{g}''_{SA} (x) = \frac{1}{\left( {x^2 + 1} \right)\sqrt {2x^2 + 2} } > 0,
\,
\forall x > 0.
\end{equation}

\item \textbf{For }$\bf{D_{SN_2 } (a,b)}$\textbf{:} We can write $D_{SN_2 } (a,b) =
b\,g_{SN_2 } \left( {a \mathord{\left/ {\vphantom {a b}} \right.
\kern-\nulldelimiterspace} b} \right)$, where
\[
g_{SN_2 } (x) = f_S (x) - f_{N_2 } (x)
 = \sqrt {\frac{x^2 + 1}{2}} - \sqrt {\frac{x + 1}{2}} \left( {\frac{\sqrt x
+ 1}{2}} \right).
\]

This gives
\begin{align}
{g}''_{SN_2 } (x) & = \frac{1}{8x^{3 / 2}\left( {x^2 + 1} \right)\left( {x +
1} \right)\sqrt {2x^2 + 2} \sqrt {2x + 2} }\times\notag\\
\label{eq36}
& \hspace{15pt} \times \left( {\begin{array}{l}
 \sqrt {2x^2 + 2} \left( {x^{3 / 2} + 1} \right)\left( {x^2 + 1} \right) +
\\
 + 8x^{3 / 2}\left( {x + 1} \right)\sqrt {2x + 2} \\
 \end{array}} \right),\,\forall x > 0.
\end{align}

\item \textbf{For }$\bf{D_{SN_1 } (a,b)}$\textbf{:} We can write $D_{SN_1 } (a,b) =
b\,g_{SN_1 } \left( {a \mathord{\left/ {\vphantom {a b}} \right.
\kern-\nulldelimiterspace} b} \right)$, where
\[
g_{SN_1 } (x) = f_S (x) - f_{N_1 } (x)
 = \sqrt {\frac{x^2 + 1}{2}} - \left( {\frac{\sqrt x + 1}{2}} \right)^2.
\]

This gives
\begin{equation}
\label{eq37}
{g}''_{SN_1 } (x) = \frac{8x^{3 / 2} + \left( {x^2 + 1} \right)\sqrt {2x^2 +
2} }{8x^{3 / 2}\left( {x^2 + 1} \right)\sqrt {2x^2 + 2} } > 0,
\,
\forall x > 0.
\end{equation}

\item \textbf{For }$\bf{D_{AN_2 } (a,b)}$\textbf{:} We can write $D_{AN_2 } (a,b) =
b\,g_{AN_2 } \left( {a \mathord{\left/ {\vphantom {a b}} \right.
\kern-\nulldelimiterspace} b} \right)$, where
\[
g_{AN_2 } (x) = f_A (x) - f_{N_2 } (x)
 = \frac{x + 1}{2} - \sqrt {\frac{x + 1}{2}} \left( {\frac{\sqrt x + 1}{2}}
\right).
\]

This gives
\begin{equation}
\label{eq38}
{g}''_{AN_2 } (x) = \frac{x^{3 / 2} + 1}{8x^{3 / 2}\left( {x + 1}
\right)\sqrt {2x + 2} } > 0,
\,
\forall x > 0.
\end{equation}

\item \textbf{For }$\bf{D_{AG} (a,b)}$\textbf{:} We can write $D_{AG} (a,b) = b\,g_{AG}
\left( {a \mathord{\left/ {\vphantom {a b}} \right.
\kern-\nulldelimiterspace} b} \right)$, where
\[
g_{AG} (x) = f_A (x) - f_G (x)
 = \frac{x + 1}{2} - \sqrt x
 = \frac{\left( {\sqrt x - 1} \right)^2}{2}.
\]

This gives
\begin{equation}
\label{eq39}
{g}''_{AG} (x) = \frac{1}{4x^{3 / 2}} > 0,
\,
\forall x > 0.
\end{equation}

\item \textbf{For }$\bf{D_{AH} (a,b)}$\textbf{:} We can write $D_{AH} (a,b) = b\,g_{AH}
\left( {a \mathord{\left/ {\vphantom {a b}} \right.
\kern-\nulldelimiterspace} b} \right)$, where
\[
g_{AH} (x) = f_A (x) - f_H (x)
 = \frac{x + 1}{2} - \frac{2x}{x + 1}
 = \frac{\left( {x - 1} \right)^2}{2\left( {x + 1} \right)}.
\]

This gives
\begin{equation}
\label{eq40}
{g}''_{AH} (x) = \frac{4}{\left( {x + 1} \right)^3} > 0,
\,
\forall x > 0.
\end{equation}

\item \textbf{For }$\bf{D_{N_2 N_1 } (a,b)}$\textbf{:} We can write $D_{N_2 N_1 } (a,b)
= b\,g_{N_2 N_1 } \left( {a \mathord{\left/ {\vphantom {a b}} \right.
\kern-\nulldelimiterspace} b} \right)$, where
\[
g_{N_2 N_1 } (x) = f_{N_2 } (x) - f_{N_1 } (x)
 = \sqrt {\frac{x^2 + 1}{2}} - \left( {\frac{\sqrt x + 1}{2}} \right)^2.
\]

This gives
\begin{equation}
\label{eq41}
{g}''_{N_2 N_1 } (x) = \frac{\left( {x + 1} \right)\sqrt {2x + 2} - \left(
{x^{3 / 2} + 1} \right)}{8x^{3 / 2}\left( {x + 1} \right)\sqrt {2x + 2} }.
\end{equation}

From (\ref{eq14}), we have $\left( {x + 1} \right) \geqslant \left( {\frac{\sqrt x
+ 1}{2}} \right)\sqrt {2x + 2} $. This gives
\begin{align}
\left( {x + 1} \right) &\sqrt {2x + 2} - \left( {x^{3 / 2} + 1} \right)\notag\\
 & = \left( {x + 1} \right)\sqrt {2x + 2} - \left( {\sqrt x + 1} \right)\left(
{x - \sqrt x + 1} \right) \notag\\
& \geqslant \left( {\frac{\sqrt x + 1}{2}} \right)\left[ {2x + 2 - 2\left( {x
- \sqrt x + 1} \right)} \right]\notag\\
& \geqslant \sqrt x \left( {\sqrt x + 1} \right) > 0, \,\forall x > 0.\notag
\end{align}

This proves that ${g}''_{N_2 N_1 } (x) > 0$, $\forall x > 0$.

\item \textbf{For }$\bf{D_{SG} (a,b)}$\textbf{:} We can write $D_{SG} (a,b) = b\,g_{SG}
\left( {a \mathord{\left/ {\vphantom {a b}} \right.
\kern-\nulldelimiterspace} b} \right)$, where
\[
g_S (x) = f_S (x) - f_G (x)
 = \sqrt {\frac{x^2 + 1}{2}} - \sqrt x .
\]

This gives
\begin{equation}
\label{eq42}
{g}''_{AH} (x) = \frac{4x^{3 / 2} + \sqrt {2x^2 + 2} \left( {x^2 + 1}
\right)}{4x^{3 / 2}\sqrt {2x^2 + 2} \left( {x^2 + 1} \right)} > 0,
\,
\forall x > 0.
\end{equation}
\end{enumerate}

The convexity of the expressions given in parts 20-27 is already given in
\cite{tan4} but we have written them again because we need them in the next
section.

\section{Sequences of Inequalities}

In this section we shall bring sequence of inequalities based on the differences arising due to (\ref{eq8}). The results given in this section are based on the applications of the following
lemma \cite{tan4}:

\begin{lemma} Let $f_1 ,f_2 :I \subset \mathbb{R}_ + \to \mathbb{R}$
be two convex functions satisfying the assumptions:

\noindent (i) $f_1 (1) = f_1 ^\prime (1) = 0$, $f_2 (1) = f_2 ^\prime (1) = 0$;\\
\noindent (ii) $f_1 $ and $f_2 $ are twice differentiable in $\mathbb{R}_ + $;\\
\noindent (iii) there exists the real constants $\alpha ,\beta $ such that $0
\leqslant \alpha < \beta $ and
\[
\alpha \leqslant \frac{f_1 ^{\prime \prime }(x)}{f_2 ^{\prime \prime }(x)}
\leqslant \beta, \, f_2 ^{\prime \prime }(x) > 0,
\]

\noindent
for all $x > 0$ then we have the inequalities:
\[
\alpha \mbox{ }\phi _{f_2 } (a,b) \leqslant \phi _{f_1 } (a,b) \leqslant
\beta \mbox{ }\phi _{f_2 } (a,b),
\]

\noindent
for all $a,b \in (0,\infty )$, where the function $\varphi _{( \cdot )}
(a,b)$ is as defined in Lemma 2.1.
\end{lemma}

\begin{remark} From the above lemma we observe that
\[
\eta (x) = \beta \,f_2 (x) - f_1 (x),
\,
\forall x > 0.
\]

\noindent This gives $\eta (1) = {\eta }'(1) = 0$. In order to obtain $\beta $, let us
consider ${\eta }''(1) = 0$. This gives $\beta = {f_1 ^{\prime \prime }(1)}
\mathord{\left/ {\vphantom {{f_1 ^{\prime \prime }(1)} {f_2 ^{\prime \prime
}(1)}}} \right. \kern-\nulldelimiterspace} {f_2 ^{\prime \prime }(1)}$. We
shall use this argument to prove the results given in the theorem below.
\end{remark}

\begin{theorem} The following sequences of inequalities hold:
\begin{align}
& \frac{1}{8}D_{P_6 P_1 } \leqslant \frac{1}{6}D_{P_6 P_2 } \leqslant D_{SA}
\leqslant \frac{1}{3}D_{SH} \leqslant \frac{1}{2}D_{AH} \leqslant\notag\\
& \hspace{10pt} \leqslant \left\{ {\begin{array}{l}
 \tfrac{4}{9}D_{P_6 N_2 } \\\\
 \left\{ {\begin{array}{l}
 \tfrac{3}{7}D_{P_6 N_3 } \\
 \tfrac{2}{5}D_{SP_4 } \\
 \end{array}} \right\} \leqslant \left\{ {\begin{array}{l}
 \tfrac{2}{5}D_{P_6 N_1 } \\
 \tfrac{2}{7}D_{P_6 P_4 } \\
 \end{array}} \right. \\
 \end{array}} \right\} \leqslant \frac{1}{3}D_{P_6 G} \leqslant \left\{
{\begin{array}{l}
 \tfrac{2}{5}D_{P_5 H} \\
 \tfrac{2}{3}D_{AP_4 } \\
 \end{array}} \right\} \leqslant \notag\\
& \hspace{20pt}  \leqslant 4D_{N_2 N_1 }  \leqslant \frac{4}{3}D_{N_2 G} \leqslant D_{AG} \leqslant 4D_{AN_2 }\leqslant \frac{2}{3}D_{P_5 G} \leqslant  \notag\\
\label{eq43}
& \hspace{30pt} \leqslant D_{P_5 N_1 } \leqslant \frac{6}{5}D_{P_5 N_3 } \leqslant \frac{4}{3}D_{P_5 N_2 }
\leqslant 2D_{P_5 A},\\\notag\\
\label{eq44}
& D_{SA} \leqslant \left\{ {\begin{array}{l}
 \tfrac{4}{5}D_{SN_2 } \\
 \tfrac{3}{4}D_{SN_3 } \\
 \end{array}} \right\} \leqslant \frac{2}{3}D_{SN_1 } \leqslant \left\{
{\begin{array}{l}
 \tfrac{1}{3}D_{P_6 G} \\
 \tfrac{1}{2}D_{SG} \\
 \end{array}} \right\} \leqslant \frac{2}{5}D_{P_5 H}
\intertext{and}
\label{eq45}
& \left\{ {\begin{array}{l}
 \tfrac{1}{8}D_{P_6 P_1 } \\
 \tfrac{2}{13}D_{P_5 P_1 } \\
 \end{array}} \right\} \leqslant \left\{ {\begin{array}{l}
 \tfrac{1}{6}D_{P_6 P_2 } \\
 \tfrac{2}{9}D_{P_5 P_2 } \\
 \end{array}} \right\} \leqslant \frac{2}{7}D_{P_5 P_3 } \leqslant
\frac{4}{9}D_{P_6 N_2 } \leqslant D_{P_6 S} \leqslant D_{AG}.
\end{align}
\end{theorem}

\begin{proof} We shall prove each part separately.
\begin{enumerate}
\item \textbf{For }$\bf{D_{P_6 P_1 } \leqslant \frac{4}{3}D_{P_6 P_2 } }$\textbf{: }Let
us consider the function $_{ }g_{P_6 P_1 \_P_6 P_2 } (x) = {{f}''_{P_6 P_1
} (x)} \mathord{\left/ {\vphantom {{{f}''_{P_6 P_1 } (x)} {{f}''_{P_6 P_2 }
(x)}}} \right. \kern-\nulldelimiterspace} {{f}''_{P_6 P_2 } (x)}$, where
${f}''_{P_6 P_1 } (x)$and ${f}''_{P_6 P_2 } (x)$ are as given by (\ref{eq23}) and
(\ref{eq22}) respectively. This gives
\[
\beta _{P_6 P_1 \_P_6 P_1 } = g_{P_6 P_1 \_P_6 P_2 } (1) = \frac{{f}''_{P_6
P_1 } (1)}{{f}''_{P_6 P_2 } (1)} = \frac{4}{3}.
\]

In order to prove this part we shall show that $\frac{4}{3}D_{P_6 P_2 } -
D_{P_6 P_1 } \geqslant 0$. We can write
\[
\frac{4}{3}D_{P_6 P_2 } - D_{P_6 P_1 } = bf_{P_6 P_2 \_P_6 P_1 } \left(
{\frac{a}{b}} \right),
\]

\noindent where
\[
f_{P_6 P_2 \_P_6 P_1 } (x)  = \frac{4}{3}f_{P_6 P_2 } (x) - f_{P_6 P_1 } (x)\notag\\
  = \frac{\left( {x + 1} \right)^2\left(
{x - 1} \right)^4}{\left( {x^3 + 1} \right)\left( {x^2 + 1} \right)}.\notag
\]

Since $f_{P_6 P_2 \_P_6 P_1 } (x) > 0$, $\forall x > 0,\,x \ne 1$, hence
proving the required result.

\item \textbf{For }$\bf{D_{P_6 P_2 } \leqslant 6D_{SA} }$\textbf{: }By\textbf{
}considering the function $g_{P_6 P_2 \_SA} (x) = {{f}''_{P_6 P_2 } (x)}
\mathord{\left/ {\vphantom {{{f}''_{P_6 P_2 } (x)} {{f}''_{SA} (x)}}}
\right. \kern-\nulldelimiterspace} {{f}''_{SA} (x)}$, where ${f}''_{P_6 P_2
} (x)$ and ${f}''_{SA} (x)$ are as given by (\ref{eq22}) and (\ref{eq35}) respectively,
we get $\beta _{P_6 P_1 \_SA} = g_{P_6 P_1 \_SA} (1) = 6$. We can write
\[
6D_{SA} - D_{P_6 P_2 } = b\,f_{SA\_P_6 P_2 } \left( {\frac{a}{b}} \right),
\]

\noindent where
\[
f_{SA\_P_6 P_2 } (x) = 6f_{SA} (x) - f_{P_6 P_2 } (x)
 = \frac{k_2 (x)}{\left( {x^2 + 1} \right)\left( {x + 1} \right)},
\]

\noindent with
\[
k_2 (x) = 3\sqrt {2x + 2} \left( {x + 1} \right)\left( {x^2 + 1} \right) -
\left( {4x^4 + 5x^3 + 6x^2 + 5x + 4} \right).
\]

Now we shall show that $k_2 (x) > 0$, $\forall x > 0,\,x \ne 1$. Let us
consider
\[
h_2 (x) = \left[ {3\sqrt {2x + 2} \left( {x + 1} \right)\left( {x^2 + 1}
\right)} \right]^2 - \left( {4x^4 + 5x^3 + 6x^2 + 5x + 4} \right)^2.
\]

After simplifications we have
\[
h_2 (x) = \left( {2x^4 + 4x^3 + 3x^2 + 4x + 2} \right)\left( {x - 1}
\right)^4.
\]

Since $h_2 (x) > 0$ giving $k_2 (x) > 0$, $\forall x > 0,\,x \ne 1$. Thus we
have $f_{P_6 P_2 \_P_6 P_1 } (x) > 0$, $\forall x > 0,\,x \ne 1$, thereby
proving the required result.

\textbf{Argument:} \textit{Let }$a$\textit{ and }$b$\textit{ two positive numbers, i.e., }$a > 0$\textit{ and }$b > 0$\textit{. If }$a^2 - b^2 > 0$\textit{, then we can conclude that }$a > b$\textit{ because }$a - b = ({a^2 - b^2)} \mathord{\left/ {\vphantom {{a^2 - b^2)} {(a + b)}}} \right. \kern-\nulldelimiterspace} {(a + b)}$\textit{. We have used this argument to prove }$k_2 (x) > 0, \forall x > 0, x \ne 1$\textit{. We shall use frequently this argument to prove the other parts of the theorem. }

\item \textbf{For $\bf{D_{SA} \leqslant \frac{1}{3}D_{SH}} $:} This result is already appearing
in (\ref{eq11}).

\item \textbf{For $\bf{D_{SH} \leqslant \frac{3}{2}D_{AH}} $:} This result is already appearing
in (\ref{eq11}).

\item \textbf{For }$\bf{D_{AH} \leqslant \frac{8}{9}D_{P_6 N_2 }} $\textbf{:
}By considering the function $g_{AH\_P_6 N_2 } (x) = {{f}''_{AH}
(x)} \mathord{\left/ {\vphantom {{{f}''_{AH} (x)} {{f}''_{P_6 N_2 } (x)}}}
\right. \kern-\nulldelimiterspace} {{f}''_{P_6 N_2 } (x)}$, we get $\beta
_{AH\_P_6 N_2 } = g_{AH\_P_6 N_2 } (1) = \frac{8}{9}$, where ${f}''_{AH}
(x)$ and ${f}''_{P_6 N_2 } (x)$ are as given by (\ref{eq40}) and (\ref{eq17})
respectively. We can write
\[
\frac{8}{9}D_{P_6 N_2 } - D_{AH} = b\,f_{P_6 N_2 \_AH} \left( {\frac{a}{b}}
\right),
\]

\noindent where
\[
f_{P_6 N_2 \_AH} (x) = \frac{8}{9}f_{P_6 N_2 } (x) - f_{AH} (x)
 = \frac{k_5 (x)}{18\left( {x + 1} \right)},
\]

\noindent with
\[
k_5 (x) = 7x^2 + 18x + 7 - 4\sqrt {2x + 2} \left( {\sqrt x + 1}
\right)\left( {x + 1} \right).
\]

Now we shall show that $k_5 (x) > 0$, $\forall x > 0,\,x \ne 1$. Let us consider
\[
h_5 (x) = \left( {7x^2 + 18x + 7} \right)^2 - \left[ {4\sqrt {2x + 2} \left(
{\sqrt x + 1} \right)\left( {x + 1} \right)} \right]^2.
\]

After simplifications we have
\[
h_5 (x) = \left( {17x^2 + 4x^{3 / 2} + 38x + 4\sqrt x + 17} \right)\left(
{\sqrt x - 1} \right)^4.
\]

Since $h_5 (x) > 0$ giving $k_5 (x) > 0$, $\forall x > 0,\,x \ne 1$. Thus we
have $f_{P_6 P_2 \_AH} (x) > 0$, $\forall x > 0,\,x \ne 1$, thereby proving
the required result.

\item \textbf{For }$\bf{D_{AH} \leqslant \frac{6}{7}D_{P_6 N_3 } }$\textbf{:
}By considering the function $g_{AH\_P_6 N_3 } (x) = {{f}''_{AH}
(x)} \mathord{\left/ {\vphantom {{{f}''_{AH} (x)} {{f}''_{P_6 N_3 } (x)}}}
\right. \kern-\nulldelimiterspace} {{f}''_{P_6 N_3 } (x)}$, we get $\beta
_{AH\_P_6 N_3 } = g_{AH\_P_6 N_3 } (1) = \frac{6}{7}$, where ${f}''_{AH}
(x)$ and ${f}''_{P_6 N_3 } (x)$ are as given by (\ref{eq40}) and (\ref{eq18})
respectively. We can write
\[
\frac{6}{7}D_{P_6 N_3 } - D_{AH} = b\,f_{P_6 N_3 \_AH} \left( {\frac{a}{b}}
\right),
\]

\noindent where
\[
f_{P_6 N_3 \_AH} (x) = \frac{6}{7}f_{P_6 N_3 } (x) - f_{AH} (x)
 = \frac{\left( {\sqrt x - 1} \right)^2}{18\left( {x + 1} \right)}.
\]

Since $f_{P_6 N_3 \_AH} (x) > 0$, $\forall x > 0,\,x \ne 1$, hence proving
the required result.

\item \textbf{For }$\bf{D_{AH} \leqslant \frac{2}{3}D_{SP_4 } }$\textbf{: }By\textbf{
}considering the function $g_{AH\_SP_4 } (x) = {{f}''_{AH} (x)}
\mathord{\left/ {\vphantom {{{f}''_{AH} (x)} {{f}''_{SP_4 } (x)}}} \right.
\kern-\nulldelimiterspace} {{f}''_{SP_4 } (x)}$, we get $\beta _{AH\_SP_4 }
= g_{AH\_SP_4 } (1) = \frac{2}{3}$, where ${f}''_{AH} (x)$ and ${f}''_{SP_4
} (x)$ are as given by (\ref{eq40}) and (\ref{eq33}) respectively. We can write
\[
\frac{2}{3}D_{SP_4 } - D_{AH} = b\,f_{SP_4 \_AH} \left( {\frac{a}{b}}
\right),
\]

\noindent where
\[
f_{SP_4 \_AH} (x) = \frac{2}{3}f_{SP_4 } (x) - f_{AH} (x)
 = \frac{k_7 (x)}{10\left( {x + 1} \right)\left( {\sqrt x + 1} \right)^2},
\]

\noindent with
\begin{align}
k_7 (x) & = 4\left( {x + 1} \right)\sqrt {2x + 2} \left( {\sqrt x + 1}
\right)^2 \notag\\
& \hspace{15pt}  - \left[ {5x^3 + 10x^{5 / 2} + 17x^2 + 17x + 10x\left( {\sqrt x - 1}
\right)^2 + 10\sqrt x + 5} \right].\notag
\end{align}

Now we shall show that $k_7 (x) > 0$, $\forall x > 0,\,x \ne 1$. Let us consider
\begin{align}
h_7 (x) & = \left[ {4\left( {x + 1} \right)\sqrt {2x + 2} \left( {\sqrt x + 1}
\right)^2} \right]^2 \notag\\
& \hspace{15pt} - \left[ {5x^3 + 10x^{5 / 2} + 17x^2 + 17x + 10x\left( {\sqrt x - 1}
\right)^2 + 10\sqrt x + 5} \right]^2.\notag
\end{align}

After simplifications we have
\[
h_7 (x) = \left( {7x^3 + 70x^{5 / 2} + 201x^2 + 340x^{3 / 2} + 201x +
70\sqrt x + 7} \right)\left( {\sqrt x - 1} \right)^6
\]

Since $h_7 (x) > 0$ giving $k_7 (x) > 0$, $\forall x > 0,\,x \ne 1$. Thus we
have $f_{SP_4 \_AH} (x) > 0$, $\forall x > 0,\,x \ne 1$, thereby proving the
required result.

\item \textbf{For }$\bf{D_{P_6 N_3 } \leqslant \frac{14}{15}D_{P_6 N_1 } }$\textbf{:
}By considering the function $g_{P_6 N_3 \_P_6 N_1 } (x) =
{{f}''_{P_6 N_3 } (x)} \mathord{\left/ {\vphantom {{{f}''_{P_6 N_3 } (x)}
{{f}''_{P_6 N_1 } (x)}}} \right. \kern-\nulldelimiterspace} {{f}''_{P_6 N_1
} (x)}$, we get $\beta _{P_6 N_3 \_P_6 N_1 } = g_{P_6 N_3 \_P_6 N_1 } (1) =
\frac{14}{15}$, where ${f}''_{P_6 N_3 } (x)$ and ${f}''_{P_6 N_1 } (x)$ are
as given by (\ref{eq18}) and (\ref{eq19}) respectively. We can write
\[
\frac{14}{15}D_{P_6 N_1 } - D_{P_6 N_3 } = b\,f_{P_6 N_1 \_P_6 N_3 } \left(
{\frac{a}{b}} \right),
\]

\noindent where
\[
f_{P_6 N_1 \_P_6 N_3 } (x) = \frac{14}{15}f_{P_6 N_1 } (x) - f_{P_6 N_3 }
(x)  = \frac{\left( {\sqrt x - 1} \right)^2}{30\left( {x + 1} \right)}.
\]

Since $f_{P_6 N_1 \_P_6 N_3 } (x) > 0$, $\forall x > 0,\,x \ne 1$, hence
proving the required result.

\item \textbf{For $\bf{D_{P_6 N_3 } \leqslant \frac{2}{3}D_{P_6 P_4 } }:$} By considering the
function $g_{P_6 N_3 \_P_6 P_4 } (x) = {{f}''_{P_6 N_3 } (x)}
\mathord{\left/ {\vphantom {{{f}''_{P_6 N_3 } (x)} {{f}''_{P_6 P_4 } (x)}}}
\right. \kern-\nulldelimiterspace} {{f}''_{P_6 P_4 } (x)}$, we get $\beta
_{P_6 N_3 \_P_6 P_4 } = g_{P_6 N_3 \_P_6 P_4 } (1) = \frac{2}{3}$, where
${f}''_{P_6 N_3 } (x)$ and ${f}''_{P_6 P_4 } (x)$ are as given by (\ref{eq18}) and
(\ref{eq21}) respectively. We can write
\[
\frac{2}{3}D_{P_6 P_4 } - D_{P_6 N_3 } = b\,f_{P_6 P_4 \_P_6 N_3 } \left(
{\frac{a}{b}} \right),
\]

\noindent where
\[
f_{P_6 P_4 \_P_6 N_3 } (x) = \frac{2}{3}f_{P_6 P_4 } (x) - f_{P_6 N_3 } (x)
 = \frac{\sqrt x \left( {\sqrt x - 1} \right)^4}{3\left( {\sqrt x + 1}
\right)^2\left( {x + 1} \right)}.
\]

Since $f_{P_6 P_4 \_P_6 N_3 } (x) > 0$, $\forall x > 0,\,x \ne 1$, hence
proving the required result.

\item \textbf{For $\bf{D_{SP_4 } \leqslant \frac{4}{5}D_{P_6 N_1 }} :$} By considering the
function $g_{SP_4 \_P_6 N_1 } (x) = {{f}''_{SP_4 } (x)} \mathord{\left/
{\vphantom {{{f}''_{SP_4 } (x)} {{f}''_{P_6 N_1 } (x)}}} \right.
\kern-\nulldelimiterspace} {{f}''_{P_6 N_1 } (x)}$, we get $\beta _{SP_4
\_P_6 N_1 } = g_{SP_4 \_P_6 N_1 } (1) = \frac{4}{5}$, where ${f}''_{SP_4 }
(x)$ and ${f}''_{P_6 N_1 } (x)$ are as given by (\ref{eq33}) and (\ref{eq19})
respectively. We can write
\[
\frac{4}{5}D_{P_6 N_1 } - D_{SP_4 } = b\,f_{P_6 N_1 \_SP_4 } \left(
{\frac{a}{b}} \right),
\]

\noindent where
\[
f_{P_6 N_1 \_SP_4 } (x) = \frac{4}{5}f_{P_6 N_1 } (x) - f_{SP_4 } (x)
 = \frac{k_{10} (x)}{4\left( {x + 1} \right)\left( {\sqrt x + 1}
\right)^2},
\]

\noindent with
\begin{align}
k_{10} (x) & = 3x^3 + 4x^{5 / 2} + 9x^2 + 4x\left( {\sqrt x - 1} \right)^2 +
9x + 4\sqrt x + 3\notag\\
&\hspace{15pt} - 2\sqrt {2x^2 + 2} \left( {x + 1} \right)\left( {\sqrt x + 1} \right)^2.\notag
\end{align}

Now we shall show that $k_{10} (x) > 0$, $\forall x > 0,\,x \ne 1$. Let us consider
\begin{align}
h_{10} (x) & = \left[ {3x^3 + 4x^{5 / 2} + 9x^2 + 4x\left( {\sqrt x - 1}
\right)^2 + 9x + 4\sqrt x + 3} \right]^2\notag\\
& \hspace{15pt} - \left[ {2\sqrt {2x^2 + 2} \left( {x + 1} \right)\left( {\sqrt x + 1}
\right)^2} \right]^2.\notag
\end{align}

After simplifications we have
\[
h_{10} (x) = \left( {\sqrt x - 1} \right)^4\left( {\begin{array}{l}
 x^3\left( {\sqrt x - 2} \right)^2 + 4x^3 + 20x^{5 / 2} + \\
 + 78x^2 + 20x^{3 / 2} + 4x + \left( {2\sqrt x - 1} \right)^2 \\
 \end{array}} \right).
\]

Since $h_{10} (x) > 0$ giving $k_{10} (x) > 0$, $\forall x > 0,\,x \ne 1$.
Thus we have $f_{P_6 N_1 \_SP_4 } (x) > 0$, $\forall x > 0,\,x \ne 1$,
thereby proving the required result.

\item \textbf{For $\bf{D_{SP_4 } \leqslant \frac{5}{7}D_{P_6 P_4 } }:$} By considering the
function $g_{SP_4 \_P_6 P_4 } (x) = {{f}''_{SP_4 } (x)} \mathord{\left/
{\vphantom {{{f}''_{SP_4 } (x)} {{f}''_{P_6 P_4 } (x)}}} \right.
\kern-\nulldelimiterspace} {{f}''_{P_6 P_4 } (x)}$, we get $\beta _{SP_4
\_P_6 P_4 } = g_{SP_4 \_P_6 P_4 } (1) = \frac{5}{7}$, where ${f}''_{SP_4 }
(x)$ and ${f}''_{P_6 P_4 } (x)$ are as given by (\ref{eq33}) and (\ref{eq21})
respectively. We can write
\[
\frac{5}{7}D_{P_6 P_4 } - D_{SP_4 } = b\,f_{P_6 P_4 \_SP_4 } \left(
{\frac{a}{b}} \right),
\]

\noindent where
\[
f_{P_6 P_4 \_SP_4 } (x) = \frac{5}{7}f_{P_6 P_4 } (x) - f_{SP_4 } (x)
 = \frac{k_{11} (x)}{14\left( {x + 1} \right)\left( {\sqrt x + 1}
\right)^2},
\]

\noindent with
\begin{align}
k_{11} (x) & = 10x^3 + 20x^{5 / 2} + 26x^2 + 26x + 20\sqrt x + 10\notag\\
& \hspace{15pt} - 7\sqrt {2x^2 + 2} \left( {\sqrt x + 1} \right)^2\left( {x + 1} \right).\notag
\end{align}

Now we shall show that $k_{11} (x) > 0$, $\forall x > 0,\,x \ne 1$. Let us consider
\begin{align}
h_{11} (x) & = \left( {10x^3 + 20x^{5 / 2} + 26x^2 + 26x + 20\sqrt x + 10}
\right)^2\notag\\
& \hspace{15pt} - \left[ {7\sqrt {2x^2 + 2} \left( {\sqrt x + 1} \right)^2\left( {x + 1}
\right)} \right]^2.\notag
\end{align}

After simplifications we have
\[
h_{11} (x) = 2\left( {\sqrt x - 1} \right)^4\left( {\begin{array}{l}
 x^4 + 8x^{7 / 2} + 94x^3 + 264x^{5 / 2} + \\
 + 386x^2 + 264x^{3 / 2} + 94x + 8\sqrt x + 1 \\
 \end{array}} \right).
\]

Since $h_{11} (x) > 0$ giving $k_{11} (x) > 0$, $\forall x > 0,\,x \ne 1$.
Thus we have $f_{P_6 P_4 \_SP_4 } (x) > 0$, $\forall x > 0,\,x \ne 1$,
thereby proving the required result.

\item \textbf{For $\bf{D_{P_6 N_2 } \leqslant \frac{3}{4}D_{P_6 G}} :$} By considering the
function $g_{P_6 N_2 \_P_6 G} (x) = {{f}''_{P_6 N_2 } (x)} \mathord{\left/
{\vphantom {{{f}''_{P_6 N_2 } (x)} {{f}''_{P_6 G} (x)}}} \right.
\kern-\nulldelimiterspace} {{f}''_{P_6 G} (x)}$, we get $\beta _{P_6 N_2
\_P_6 G} = g_{P_6 N_2 \_P_6 G} (1) = \frac{3}{4}$, where ${f}''_{P_6 N_2 }
(x)$ and ${f}''_{P_6 G} (x)$ are as given by (\ref{eq17}) and (\ref{eq20}) respectively.
We can write
\[
\frac{3}{4}D_{P_6 G} - D_{P_6 N_2 } = b\,f_{P_6 G\_P_6 N_2 } \left(
{\frac{a}{b}} \right),
\]

\noindent where
\[
f_{P_6 G\_P_6 N_2 } (x) = \frac{3}{4}f_{P_6 G} (x) - f_{P_6 N_2 } (x)
 = \frac{k_{12} (x)}{4\left( {x + 1} \right)},
\]

\noindent with
\[
k_{12} (x) = \sqrt {2x + 2} \left( {\sqrt x + 1} \right)\left( {x + 1}
\right) - \left( {x^2 + 1 + 3x^{3 / 2} + 3\sqrt x } \right).
\]

Now we shall show that $k_5 (x) > 0$, $\forall x > 0,\,x \ne 1$. Let us consider
\[
h_{12} (x) = \left[ {\sqrt {2x + 2} \left( {\sqrt x + 1} \right)\left( {x +
1} \right)} \right]^2 - \left( {x^2 + 1 + 3x^{3 / 2} + 3\sqrt x }
\right)^2.
\]

After simplifications we have
\[
h_{12} (x) = \left( {\sqrt x - 1} \right)^4\left( {x^2 + 2x^{3 / 2} + x +
2\sqrt x + 1} \right).
\]

Since $h_{12} (x) > 0$ giving $k_{12} (x) > 0$, $\forall x > 0,\,x \ne 1$.
Thus we have $f_{P_6 G\_P_6 N_2 } (x) > 0$, $\forall x > 0,\,x \ne 1$,
thereby proving the required result.

\item \textbf{For }$\bf{D_{P_6 N_1 } \leqslant \frac{5}{6}D_{P_6 G}} $\textbf{:
}By considering the function $g_{P_6 N_1 \_P_6 G} (x) = {{f}''_{P_6
N_1 } (x)} \mathord{\left/ {\vphantom {{{f}''_{P_6 N_1 } (x)} {{f}''_{P_6 G}
(x)}}} \right. \kern-\nulldelimiterspace} {{f}''_{P_6 G} (x)}$, we get
$\beta _{P_6 N_1 \_P_6 G} = g_{P_6 N_1 \_P_6 G} (1) = \frac{5}{6}$, where
${f}''_{P_6 N_1 } (x)$ and ${f}''_{P_6 G} (x)$ are as given by (\ref{eq19}) and
(\ref{eq20}) respectively. We can write
\[
\frac{5}{6}D_{P_6 G} - D_{P_6 N_1 } = b\,f_{P_6 G\_P_6 N_1 } \left(
{\frac{a}{b}} \right),
\]

\noindent where
\[
f_{P_6 G\_P_6 N_1 } (x) = \frac{5}{6}f_{P_6 G} (x) - f_{P_6 N_1 } (x)
 = \frac{\left( {\sqrt x - 1} \right)^2}{12\left( {x + 1} \right)}.
\]

Since $f_{P_6 G\_P_6 N_1 } (x) > 0$, $\forall x > 0,\,x \ne 1$, hence
proving the required result.

\item  \textbf{For }$\bf{D_{P_6 P_4 } \leqslant \frac{7}{6}D_{P_6 G}} $\textbf{:
}By considering the function $g_{P_6 P_4 \_P_6 G} (x) = {{f}''_{P_6
P_4 } (x)} \mathord{\left/ {\vphantom {{{f}''_{P_6 P_4 } (x)} {{f}''_{P_6 G}
(x)}}} \right. \kern-\nulldelimiterspace} {{f}''_{P_6 G} (x)}$, we get
$\beta _{P_6 P_4 \_P_6 G} = g_{P_6 P_4 \_P_6 G} (1) = \frac{7}{6}$, where
${f}''_{P_6 P_4 } (x)$ and ${f}''_{P_6 G} (x)$ are as given by (\ref{eq21}) and
(\ref{eq20}) respectively. We can write
\[
\frac{7}{6}D_{P_6 G} - D_{P_6 P_4 } = b\,f_{P_6 G\_P_6 P_4 } \left(
{\frac{a}{b}} \right),
\]

\noindent where
\[
f_{P_6 G\_P_6 P_4 } (x) = \frac{7}{6}f_{P_6 G} (x) - f_{P_6 P_4 } (x)
 = \frac{\left( {\sqrt x - 1} \right)^4\left[ {\left( {\sqrt x - 1}
\right)^2 + \sqrt x } \right]}{12\left( {x + 1} \right)}.
\]

Since $f_{P_6 G\_P_6 P_4 } (x) > 0$, $\forall x > 0,\,x \ne 1$, hence
proving the required result.

\item \textbf{For }$\bf{D_{P_6 G} \leqslant \frac{6}{5}D_{P_5 H}} $\textbf{:
}By considering the function $g_{P_6 G\_P_5 H} (x) = {{f}''_{P_6 G}
(x)} \mathord{\left/ {\vphantom {{{f}''_{P_6 G} (x)} {{f}''_{P_5 H} (x)}}}
\right. \kern-\nulldelimiterspace} {{f}''_{P_5 H} (x)}$, we get $\beta _{P_6
G\_P_5 H} = g_{P_6 G\_P_5 H} (1) = \frac{6}{5}$, where ${f}''_{P_6 G} (x)$
and ${f}''_{P_5 H} (x)$ are as given by (\ref{eq20}) and (\ref{eq29}) respectively. We
can write
\[
\frac{6}{5}D_{P_5 H} - D_{P_6 G} = b\,f_{P_5 H\_P_6 G} \left( {\frac{a}{b}}
\right),
\]

\noindent where
\[
f_{P_5 H\_P_6 G} (x) = \frac{6}{5}f_{P_5 H} (x) - f_{P_6 G} (x)
 = \frac{\left( {\sqrt x - 1} \right)^4\left[ {\left( {\sqrt x - 1}
\right)^2 + \sqrt x } \right]}{5\left( {x + 1} \right)}.
\]

Since $f_{P_5 H\_P_6 G} (x) > 0$, $\forall x > 0,\,x \ne 1$, hence proving
the required result.

\item  \textbf{For }$\bf{D_{P_6 G} \leqslant 2D_{AP_4 } }$\textbf{: }By\textbf{
}considering the function $g_{P_6 G\_AP_4 } (x) = {{f}''_{P_6 G} (x)}
\mathord{\left/ {\vphantom {{{f}''_{P_6 G} (x)} {{f}''_{AP_4 } (x)}}}
\right. \kern-\nulldelimiterspace} {{f}''_{AP_4 } (x)}$, we get $\beta _{P_6
G\_AP_4 } = g_{P_6 G\_AP_4 } (1) = 2$, where ${f}''_{P_6 G} (x)$ and
${f}''_{AP_4 } (x)$ are as given by (\ref{eq20}) and (\ref{eq34}) respectively. We can
write
\[
2D_{AP_4 } - D_{P_6 G} = b\,f_{AP_4 \_P_6 G} \left( {\frac{a}{b}} \right),
\]

\noindent where
\[
f_{AP_4 \_P_6 G} (x) = 2f_{AP_4 } (x) - f_{P_6 G} (x)
 = \frac{\sqrt x \left( {\sqrt x - 1} \right)^4}{\left( {\sqrt x + 1}
\right)^2\left( {x + 1} \right)}.
\]

Since $f_{AP_4 \_P_6 G} (x) > 0$, $\forall x > 0,\,x \ne 1$, hence proving
the required result.

\item  \textbf{For }$\bf{D_{P_5 H} \leqslant 10D_{N_2 N_1 } }$\textbf{: }By\textbf{
}considering the function $g_{P_5 H\_N_2 N_1 } (x) = {{f}''_{P_5 H} (x)}
\mathord{\left/ {\vphantom {{{f}''_{P_5 H} (x)} {{f}''_{N_2 N_1 } (x)}}}
\right. \kern-\nulldelimiterspace} {{f}''_{N_2 N_1 } (x)}$, we get $\beta
_{P_5 H\_N_2 N_1 } = g_{P_5 H\_N_2 N_1 } (1) = 10$, where ${f}''_{P_5 H}
(x)$ and ${f}''_{N_2 N_1 } (x)$ are as given by (\ref{eq29}) and (\ref{eq41})
respectively. We can write
\[
10D_{N_2 N_1 } - D_{P_5 H} = b\,f_{N_2 N_1 \_P_5 H} \left( {\frac{a}{b}}
\right),
\]

\noindent where
\[
f_{N_2 N_1 \_P_5 H} (x) = 10f_{N_2 N_1 } (x) - f_{P_5 H} (x)
 = \frac{k_{17} (x)}{2\left( {\sqrt x + 1} \right)^2\left( {x + 1}
\right)},
\]

\noindent with
\begin{align}
k_{17} (x) & = 5\sqrt {2x + 2} \left( {\sqrt x + 1} \right)^3\left( {x + 1}
\right)\notag\\
& \hspace{15pt} - \left( {7x^3 + 20x^{5 / 2} + 37x^2 + 32x^{3 / 2} + 37x + 20\sqrt x + 7}
\right).\notag
\end{align}

Now we shall show that $k_{17} (x) > 0$, $\forall x > 0,\,x \ne 1$. Let us consider
\begin{align}
h_{17} (x) & = \left[ {5\sqrt {2x + 2} \left( {\sqrt x + 1} \right)^3\left( {x
+ 1} \right)} \right]^2 \notag\\
& \hspace{15pt} - \left( {7x^3 + 20x^{5 / 2} + 37x^2 + 32x^{3 / 2} + 37x + 20\sqrt x + 7}
\right)^2. \notag
\end{align}

After simplifications we have
\[
h_{17} (x) = \left( {\sqrt x - 1} \right)^4\left( {\begin{array}{l}
 x^4 + 24x^{7 / 2} + 72x^3 + 120x^{5 / 2} + \\
 + 126x^2 + 120x^{3 / 2} + 72x + 24\sqrt x + 1 \\
 \end{array}} \right).
\]

Since $h_{17} (x) > 0$ giving $k_{17} (x) > 0$, $\forall x > 0,\,x \ne 1$.
Thus we have $f_{N_2 N_1 \_P_5 H} (x) > 0$, $\forall x > 0,\,x \ne 1$,
thereby proving the required result.

\item \textbf{For $\bf{ D_{AP_4 } \leqslant 6D_{N_2 N_1 }} :$} By considering the function
$g_{AP_4 \_N_2 N_1 } (x) = {{f}''_{AP_4 } (x)} \mathord{\left/ {\vphantom
{{{f}''_{AP_4 } (x)} {{f}''_{N_2 N_1 } (x)}}} \right.
\kern-\nulldelimiterspace} {{f}''_{N_2 N_1 } (x)}$, we get $\beta _{AP_4
\_N_2 N_1 } = g_{AP_4 \_N_2 N_1 } (1) = 6$, where ${f}''_{AP_4 } (x)$ and
${f}''_{N_2 N_1 } (x)$ are as given by (\ref{eq34}) and (\ref{eq41}) respectively. We
can write
\[
6D_{N_2 N_1 } - D_{AP_4 } = b\,f_{N_2 N_1 \_AP_4 } \left( {\frac{a}{b}}
\right),
\]

\noindent where
\[
f_{N_2 N_1 \_AP_4 } (x) = 6f_{N_2 N_1 } (x) - f_{AP_4 } (x)
 = \frac{k_{18} (x)}{2\left( {\sqrt x + 1} \right)^2},
\]

\noindent with
\[
k_{18} (x) = 3\sqrt {2x + 2} \left( {\sqrt x + 1} \right)^3
 - 2\left( {2x^2 + 7x^{3 / 2} + 6x + 7\sqrt x + 2} \right).
\]

Now we shall show that $k_{18} (x) > 0$, $\forall x > 0,\,x \ne 1$. Let us consider
\[
h_{18} (x) = \left[ {3\sqrt {2x + 2} \left( {\sqrt x + 1} \right)^3}
\right]^2
 - \left( {4x^2 + 14x^{3 / 2} + 12x + 14\sqrt x + 4} \right)^2
\]

After simplifications we have
\[
h_{18} (x) = 2\left( {\sqrt x - 1} \right)^4\left( {x^2 + 2x^{3 / 2} +
2\sqrt x + 1} \right).
\]

Since $h_{18} (x) > 0$ giving $k_{18} (x) > 0$, $\forall x > 0,\,x \ne 1$.
Thus we have $f_{N_2 N_1 \_AP_4 } (x) > 0$, $\forall x > 0,\,x \ne 1$,
thereby proving the required result.

\item \textbf{For $\bf{ D_{N_2 N_1 } \leqslant \frac{1}{3}D_{N_2 G} }:$} This result is already
appearing in (\ref{eq11}).

\item \textbf{For  $\bf{D_{N_2 G} \leqslant \frac{3}{4}D_{AG}} :$} This result is already appearing in
(\ref{eq11}).

\item \textbf{For $\bf{D_{AG} \leqslant 4D_{AN_2 } }:$} This result is already appearing in (\ref{eq11}).

\item \textbf{For $\bf{D_{AN_2 } \leqslant \frac{1}{6}D_{P_5 G}} :$} By considering the
function $g_{AN_2 \_P_5 G} (x) = {{f}''_{AN_2 } (x)} \mathord{\left/
{\vphantom {{{f}''_{AN_2 } (x)} {{f}''_{P_5 G} (x)}}} \right.
\kern-\nulldelimiterspace} {{f}''_{P_5 G} (x)}$, we get $\beta _{AN_2 \_P_5
G} = g_{AN_2 \_P_5 G} (1) = \frac{1}{6}$, where ${f}''_{AN_2 } (x)$ and
${f}''_{P_5 G} (x)$ are as given by (\ref{eq38}) and (\ref{eq28}) respectively. We can
write
\[
\frac{3}{4}D_{P_5 G} - D_{AN_2 } = b\,f_{AN_2 \_P_5 G} \left( {\frac{a}{b}}
\right),
\]

\noindent where
\[
f_{AN_2 \_P_5 G} (x) = \frac{1}{6}f_{P_5 G} (x) - f_{AN_2 } (x)
 = \frac{k_{22} (x)}{12\left( {\sqrt x + 1} \right)},
\]

\noindent
with $k_{22} (x) = k_{18} (x) > 0$, $\forall x > 0,\,x \ne 1$. Thus we have
$f_{N_2 N_1 \_AP_4 } (x) > 0$, $\forall x > 0,\,x \ne 1$, thereby proving
the required result.

\item \textbf{For }$\bf{D_{P_5 G} \leqslant \frac{3}{2}D_{P_5 N_1 } }$\textbf{:
}By considering the function $g_{P_5 G\_P_5 N_1 } (x) = {{f}''_{P_5
G} (x)} \mathord{\left/ {\vphantom {{{f}''_{P_5 G} (x)} {{f}''_{P_5 N_1 }
(x)}}} \right. \kern-\nulldelimiterspace} {{f}''_{P_5 N_1 } (x)}$, we get
$\beta _{P_5 G\_P_5 N_1 } = g_{P_5 G\_P_5 N_1 } (1) = \frac{3}{2}$, where
${f}''_{P_5 G} (x)$ and ${f}''_{P_5 N_1 } (x)$ are as given by (\ref{eq28}) and
(\ref{eq27}) respectively. We can write
\[
\frac{3}{2}D_{P_5 N_1 } - D_{P_5 G} = b\,f_{P_5 G\_P_5 N_1 } \left(
{\frac{a}{b}} \right),
\]

\noindent where
\[
f_{P_5 G\_P_5 N_1 } (x) = \frac{3}{2}f_{P_5 G} (x) - f_{P_5 N_1 } (x)
 = \frac{\left( {\sqrt x - 1} \right)^2}{8\left( {x + 1} \right)}.
\]

Since $f_{P_5 G\_P_5 N_1 } (x) > 0$, $\forall x > 0,\,x \ne 1$, hence
proving the required result.

\item  \textbf{For }$\bf{D_{P_5 N_1 } \leqslant \frac{6}{5}D_{P_5 N_3 }}$\textbf{:
}By considering the function $g_{P_5 N_1 \_P_5 N_3 } (x) =
{{f}''_{P_5 N_1 } (x)} \mathord{\left/ {\vphantom {{{f}''_{P_5 N_1 } (x)}
{{f}''_{P_5 N_3 } (x)}}} \right. \kern-\nulldelimiterspace} {{f}''_{P_5 N_3
} (x)}$, we get $\beta _{P_5 N_1 \_P_5 N_3 } = g_{P_5 N_1 \_P_5 N_3 } (1) =
\frac{6}{5}$, where ${f}''_{P_5 N_1 } (x)$ and ${f}''_{P_5 N_3 } (x)$ are as
given by (\ref{eq27}) and (\ref{eq26}) respectively. We can write
\[
\frac{6}{5}D_{P_5 N_3 } - D_{P_5 N_1 } = b\,f_{P_5 N_3 \_P_5 N_1 } \left(
{\frac{a}{b}} \right),
\]

\noindent where
\[
f_{P_5 N_3 \_P_5 N_1 } (x) = \frac{6}{5}f_{P_5 N_3 } (x) - f_{P_5 N_1 } (x)
 = \frac{\left( {\sqrt x - 1} \right)^2}{20\left( {x + 1} \right)}.
\]

Since $f_{P_5 N_3 \_P_5 N_1 } (x) > 0$, $\forall x > 0,\,x \ne 1$, hence
proving the required result.

\item  \textbf{For $\bf{D_{P_5 N_3 } \leqslant \frac{10}{9}D_{P_5 N_2 } }: $} By considering the
function $g_{P_5 N_3 \_P_5 N_2 } (x) = {{f}''_{P_5 N_3 } (x)}
\mathord{\left/ {\vphantom {{{f}''_{P_5 N_3 } (x)} {{f}''_{P_5 N_2 } (x)}}}
\right. \kern-\nulldelimiterspace} {{f}''_{P_5 N_2 } (x)}$, we get $\beta
_{P_5 N_3 \_P_5 N_2 } = g_{P_5 N_3 \_P_5 N_2 } (1) = \frac{10}{9}$, where
${f}''_{P_5 N_3 } (x)$ and ${f}''_{P_5 N_2 } (x)$ are as given by (\ref{eq26}) and
(\ref{eq25}) respectively. We can write
\[
\frac{10}{9}D_{P_5 N_2 } - D_{P_5 N_3 } = b\,f_{P_5 N_2 \_P_5 N_3 } \left(
{\frac{a}{b}} \right),
\]

\noindent where
\[
f_{P_5 N_2 \_P_5 N_3 } (x) = \frac{10}{9}f_{P_5 N_2 } (x) - f_{P_5 N_3 } (x)
 = \frac{k_{25} (x)}{18\left( {\sqrt x + 1} \right)^2},
\]

\noindent with
\[
k_{25} (x) = 2\left( {4x^2 + 9x^{3 / 2} + 14x + 9\sqrt x + 4} \right) -
5\sqrt {2x + 2} \left( {\sqrt x + 1} \right)^3.
\]

Now we shall show that $k_{25} (x) > 0$, $\forall x > 0,\,x \ne 1$. Let us consider
\[
h_{25} (x) = \left[ {2\left( {4x^2 + 9x^{3 / 2} + 14x + 9\sqrt x + 4}
\right)} \right]^2 - \left[ {5\sqrt {2x + 2} \left( {\sqrt x + 1} \right)^3}
\right]^2.
\]

After simplifications we have
\[
h_{25} (x) = 2\left( {\sqrt x - 1} \right)^4\left( {7x^2 + 22x^{3 / 2} + 32x
+ 22\sqrt x + 7} \right).
\]

Since $h_{25} (x) > 0$ giving $k_{25} (x) > 0$, $\forall x > 0,\,x \ne 1$.
Thus we have $f_{P_5 N_2 \_P_5 N_3 } (x) > 0$, $\forall x > 0,\,x \ne 1$,
thereby proving the required result.

\item  \textbf{For $\bf{D_{P_5 N_2 } \leqslant \frac{3}{2}D_{P_5 A} }:$} By considering the
function $g_{P_5 N_2 \_P_5 A} (x) = {{f}''_{P_5 N_2 } (x)} \mathord{\left/
{\vphantom {{{f}''_{P_5 N_2 } (x)} {{f}''_{P_5 A} (x)}}} \right.
\kern-\nulldelimiterspace} {{f}''_{P_5 A} (x)}$, we get $\beta _{P_5 N_2
\_P_5 A} = g_{P_5 N_2 \_P_5 A} (1) = \frac{3}{2}$, where ${f}''_{P_5 N_2 }
(x)$ and ${f}''_{P_5 A} (x)$ are as given by (\ref{eq25}) and (\ref{eq24}) respectively.
We can write
\[
\frac{3}{2}D_{P_5 A} - D_{P_5 N_2 } = b\,f_{P_5 A\_P_5 N_2 } \left(
{\frac{a}{b}} \right),
\]

\noindent where
\[
f_{P_5 A\_P_5 N_2 } (x) = \frac{3}{2}f_{P_5 A} (x) - f_{P_5 N_2 } (x)
 = \frac{k_{26} (x)}{4\left( {\sqrt x + 1} \right)^2},
\]

\noindent with
\[
k_{26} (x) = \sqrt {2x + 2} \left( {\sqrt x + 1} \right)^3 - \left( {x + 1}
\right)\left( {x + 6\sqrt x + 1} \right).
\]

Let us consider
\[
h_{26} (x) = \left[ {\sqrt {2x + 2} \left( {\sqrt x + 1} \right)^3}
\right]^2 - \left[ {\left( {x + 1} \right)\left( {x + 6\sqrt x + 1} \right)}
\right]^2.
\]

After simplifications we have
\[
h_{26} (x) = \left( {\sqrt x - 1} \right)^4\left( {x + 1} \right)\left( {x +
4\sqrt x + 1} \right).
\]

Since $h_{26} (x) > 0$ giving $k_{26} (x) > 0$, $\forall x > 0,\,x \ne 1$.
Thus we have $f_{P_5 A\_P_5 N_2 } (x) > 0$, $\forall x > 0,\,x \ne 1$,
thereby proving the required result.

Parts 1-26 prove the sequences of inequalities appearing in (\ref{eq43}).

\item \textbf{For $\bf{D_{SA} \leqslant \frac{4}{5}D_{SN_2 } }:$} This result is already appearing in
(\ref{eq11}).

\item \textbf{For $\bf{D_{SA} \leqslant \frac{3}{4}D_{SN_3 }} :$} This result is already appearing in
(\ref{eq11}).

\item  \textbf{For  $\bf{D_{SN_2 } \leqslant \frac{5}{6}D_{SN_1}}:$}  By considering the
function $g_{SN_2 \_SN_1 } (x) = {{f}''_{SN_2 } (x)} \mathord{\left/
{\vphantom {{{f}''_{SN_2 } (x)} {{f}''_{SN_1 } (x)}}} \right.
\kern-\nulldelimiterspace} {{f}''_{SN_1 } (x)}$, we get $\beta _{SN_2 \_SN_1
} = g_{SN_2 \_SN_1 } (1) = \frac{5}{6}$, where ${f}''_{SN_2 } (x)$ and
${f}''_{SN_1 } (x)$ are as given by (\ref{eq36}) and (\ref{eq37}) respectively. We can
write
\[
\frac{5}{6}D_{SN_1 } - D_{SN_2 } = b\,f_{SN_1 \_SN_2 } \left( {\frac{a}{b}}
\right),
\]

\noindent where
\[
f_{SN_1 \_SN_2 } (x) = \frac{5}{6}f_{SN_1 } (x) - f_{SN_2 } (x)
 = \frac{1}{24}\times k_{29} (x),
\]

\noindent with
\[
k_{29} (x) = 6\left( {\sqrt x + 1} \right)\sqrt {2x + 2} - \left[ {2\sqrt
{2x^2 + 2} + 5\left( {\sqrt x + 1} \right)^2} \right].
\]

In order to prove $k_{29} (x) > 0$, $\forall x > 0,\,x \ne 1$, here we shall
apply twice the argument given in part 2. Let us consider
\begin{align}
h_{29} (x) & = \left[ {6\left( {\sqrt x + 1} \right)\sqrt {2x + 2} } \right]^2
- 2\left[ {\sqrt {2x^2 + 2} + 5\left( {\sqrt x + 1} \right)^2} \right]^2.\notag\\
& = 39x^2 + 41x^{3 / 2} + 3\sqrt x \left( {\sqrt x - 1} \right)^2 + 41\sqrt x
+ 39\notag\\
& \hspace{25pt}  - 20\left( {\sqrt x + 1} \right)^2\sqrt {2x^2 + 2}.\notag
\end{align}

Applying again over $h_{29} (x)$ the argument given in part 2, we get
\begin{align}
h_{29a} (x) & = \left[ {39x^2 + 41x^{3 / 2} + 3\sqrt x \left( {\sqrt x - 1}
\right)^2 + 41\sqrt x + 39} \right]^2\notag\\
& \hspace{25pt} - \left[ {20\left( {\sqrt x + 1} \right)^2\sqrt {2x^2 + 2} } \right]^2\notag\\
& = \left( {\sqrt x - 1} \right)^4\left( {721x^2 + 3116x^{3 / 2} + 4806x +
3116\sqrt x + 721} \right)^2.\notag
\end{align}

Since $h_{29a} (x) > 0$ giving $h_{29} (x) > 0$, $\forall x > 0,\,x \ne 1$
and consequently, we have $k_{29} (x) > 0$, $\forall x > 0,\,x \ne 1$. Thus
we have $f_{SN_1 \_SN_2 } (x) > 0$, $\forall x > 0,\,x \ne 1$, thereby
proving the required result.

\item \textbf{For $\bf{D_{SN_3 } \leqslant \frac{8}{9}D_{SN_1 }} :$} This result is already appearing
in (\ref{eq11}).

\item \textbf{For $\bf{D_{SN_1 } \leqslant \frac{1}{2}D_{P_6 G} }:$} By considering the
function $g_{SN_1 \_P_6 G} (x) = {{f}''_{SN_1 } (x)} \mathord{\left/
{\vphantom {{{f}''_{SN_1 } (x)} {{f}''_{P_6 G} (x)}}} \right.
\kern-\nulldelimiterspace} {{f}''_{P_6 G} (x)}$, we get $\beta _{SN_1 \_P_6
G} = g_{SN_1 \_P_6 G} (1) = \frac{1}{2}$, where ${f}''_{SN_1 } (x)$ and
${f}''_{P_6 G} (x)$ are as given by (\ref{eq37}) and (\ref{eq20}) respectively. We can
write
\[
\frac{1}{2}D_{P_6 G} - D_{SN_1 } = b\,f_{P_6 G\_SN_1 } \left( {\frac{a}{b}}
\right),
\]

\noindent where
\[
f_{P_6 G\_SN_1 } (x) = \frac{1}{2}f_{P_6 G} (x) - f_{SN_1 } (x)
 = \frac{k_{31} (x)}{4\left( {x + 1} \right)},
\]

\noindent with
\[
k_{31} (x) = 3x^2 + 2x + 3 - 2\left( {x + 1} \right)\sqrt {2x^2 + 2} .
\]

Now we shall show that $k_{31} (x) > 0$, $\forall x > 0,\,x \ne 1$. Let us
consider
\[
h_{31} (x) = \left( {3x^2 + 2x + 3} \right)^2 - \left[ {2\left( {x + 1}
\right)\sqrt {2x^2 + 2} } \right]^2
\]

After simplifications we have
\[
h_{31} (x) = \left( {x - 1} \right)^4.
\]

Since $h_{31} (x) > 0$, giving $k_{31} (x) > 0$, $\forall x > 0,\,x \ne 1$.
Thus we have $f_{P_6 G\_SN_1 } (x) > 0$, $\forall x > 0,\,x \ne 1$, thereby
proving the required result.

\item \textbf{For $\bf{D_{SN_1 } \leqslant \frac{3}{4}D_{SG} }:$} This result is already appearing in
(\ref{eq11}).

\item \textbf{For $\bf{D_{SG} \leqslant \frac{4}{5}D_{P_5 H} }:$} By considering the function
$g_{SG\_P_5 H} (x) = {{f}''_{SG} (x)} \mathord{\left/ {\vphantom
{{{f}''_{SG} (x)} {{f}''_{P_5 H} (x)}}} \right. \kern-\nulldelimiterspace}
{{f}''_{P_5 H} (x)}$, we get $\beta _{SG\_P_5 H} = g_{SG\_P_5 H} (1) =
\frac{4}{5}$, where ${f}''_{SG} (x)$ and ${f}''_{P_5 H} (x)$ are as given by
(\ref{eq42}) and (\ref{eq29}) respectively. We can write
\[
\frac{4}{5}D_{P_5 H} - D_{SG} = b\,f_{P_5 H\_SG} \left( {\frac{a}{b}}
\right),
\]

\noindent where
\[
f_{P_5 H\_SG} (x) = \frac{4}{5}f_{P_5 H} (x) - f_{SG} (x)
 = \frac{k_{33} (x)}{4\left( {x + 1} \right)},
\]

\noindent with
\begin{align}
k_{33} (x) & = 2\left( {4x^3 + 5x^{5 / 2} + 11x^2 + 3x\left( {\sqrt x - 1}
\right)^2 + 11x + 5\sqrt x + 4} \right)\notag\\
& \hspace{20pt} - 5\left( {x + 1} \right)\left( {\sqrt x + 1} \right)^2\sqrt {2x^2 + 2}.\notag
\end{align}

Now we shall show that $k_{33} (x) > 0$, $\forall x > 0,\,x \ne 1$. Let us consider
\begin{align}
h_{33} (x) & = \left[ {2\left( {4x^3 + 5x^{5 / 2} + 11x^2 + 3x\left( {\sqrt x
- 1} \right)^2 + 11x + 5\sqrt x + 4} \right)} \right]^2 \notag\\
& \hspace{20pt} - \left[ {5\left( {x + 1} \right)\left( {\sqrt x + 1} \right)^2\sqrt {2x^2
+ 2} } \right]^2. \notag
\end{align}

After simplifications we have
\[
h_{33} (x) = 2\left( {\sqrt x - 1} \right)^4\left( {\begin{array}{l}
 7x^4 + 8x^{7 / 2} + 64x^3 + 120x^{5 / 2} + \\
 + 242x^2 + 120x^{3 / 2} + 64x + 8\sqrt x + 7 \\
 \end{array}} \right).
\]

Since $h_{33} (x) > 0$, giving $k_{33} (x) > 0$, $\forall x > 0,\,x \ne 1$.
Thus we have $f_{P_5 H\_SG} (x) > 0$, $\forall x > 0,\,x \ne 1$, thereby
proving the required result.

Parts 27-33 prove the sequences of inequalities appearing in (\ref{eq44}).

\item \textbf{For $\bf{D_{P_6 P_1 } \leqslant \frac{16}{9}D_{P_5 P_2 } }:$} By considering the
function $g_{P_6 P_1 \_P_5 P_2 } (x) = {{f}''_{P_6 P_1 } (x)}
\mathord{\left/ {\vphantom {{{f}''_{P_6 P_1 } (x)} {{f}''_{P_5 P_2 } (x)}}}
\right. \kern-\nulldelimiterspace} {{f}''_{P_5 P_2 } (x)}$, we get $\beta
_{P_6 P_1 \_P_5 P_2 } = g_{P_6 P_1 \_P_5 P_2 } (1) = \frac{16}{9}$, where
${f}''_{P_6 P_1 } (x)$ and ${f}''_{P_5 P_2 } (x)$ are as given by (\ref{eq23}) and
(\ref{eq31}) respectively. We can write
\[
\frac{16}{9}D_{P_5 P_2 } - D_{P_6 P_1 } = b\,f_{P_5 P_2 \_P_6 P_1 } \left(
{\frac{a}{b}} \right),
\]

\noindent where
\begin{align}
f_{P_5 P_2 \_P_6 P_1 } (x) & = \frac{16}{9}f_{P_5 P_2 } (x) - f_{P_6 P_1 }
(x) \notag\\
& = \frac{\left( {\sqrt x - 1} \right)^4\left( {\begin{array}{l}
 7x^5 + 10x^{9 / 2} + 23x^4 + 64x^{7 / 2} + 142x^3 + \\
 + 180x^{5 / 2} + 142x^2 + 64x^{3 / 2} + 23x + 10\sqrt x + 7 \\
 \end{array}} \right)}{9\left( {x^2 + 1} \right)\left( {x^3 + 1}
\right)\left( {\sqrt x + 1} \right)^2}.\notag
\end{align}

Since $f_{P_5 P_2 \_P_6 P_1 } (x) > 0$, $\forall x > 0,\,x \ne 1$, hence
proving the required result.

\item \textbf{For $\bf{D_{P_5 P_1 } \leqslant \frac{13}{12}D_{P_6 P_2 }} :$} By considering
the function $g_{P_5 P_1 \_P_6 P_2 } (x) = {{f}''_{P_5 P_1 } (x)}
\mathord{\left/ {\vphantom {{{f}''_{P_5 P_1 } (x)} {{f}''_{P_6 P_2 } (x)}}}
\right. \kern-\nulldelimiterspace} {{f}''_{P_6 P_2 } (x)}$, we get $\beta
_{P_5 P_1 \_P_6 P_2 } = g_{P_5 P_1 \_P_6 P_2 } (1) = \frac{13}{12}$, where
${f}''_{P_5 P_1 } (x)$ and ${f}''_{P_6 P_2 } (x)$ are as given by (\ref{eq31}) and
(\ref{eq22}) respectively. We can write
\[
\frac{13}{12}D_{P_6 P_2 } - D_{P_6 P_1 } = b\,f_{P_6 P_2 \_P_5 P_1 } \left(
{\frac{a}{b}} \right),
\]

\noindent where
\begin{align}
f_{P_6 P_2 \_P_5 P_1 } (x) & = \frac{13}{12}f_{P_6 P_2 } (x) - f_{P_5 P_1 }
(x) \notag\\
& = \frac{\left( {\sqrt x - 1} \right)^4\left( {\begin{array}{l}
 x^5 + 30x^{9 / 2} + 89x^4 + 152x^{7 / 2} + 181x^3 + \\
 + 190x^{5 / 2} + 181x^2 + 152x^{3 / 2} + 89x + 30\sqrt x + 1 \\
 \end{array}} \right)}{12\left( {x^2 + 1} \right)\left( {x^3 + 1}
\right)\left( {\sqrt x + 1} \right)^2}. \notag
\end{align}

Since $f_{P_6 P_2 \_P_5 P_1 } (x) > 0$, $\forall x > 0,\,x \ne 1$, hence
proving the required result.

\item \textbf{For $\bf{D_{P_5 P_1 } \leqslant \frac{13}{9}D_{P_5 P_2 }}:$} By considering the
function $g_{P_5 P_1 \_P_5 P_2 } (x) = {{f}''_{P_5 P_1 } (x)}
\mathord{\left/ {\vphantom {{{f}''_{P_5 P_1 } (x)} {{f}''_{P_5 P_2 } (x)}}}
\right. \kern-\nulldelimiterspace} {{f}''_{P_5 P_2 } (x)}$, we get $\beta
_{P_5 P_1 \_P_5 P_2 } = g_{P_5 P_1 \_P_5 P_2 } (1) = \frac{13}{9}$, where
${f}''_{P_5 P_1 } (x)$ and ${f}''_{P_5 P_2 } (x)$ are as given by (\ref{eq32}) and
(\ref{eq31}) respectively. We can write
\[
\frac{13}{9}D_{P_5 P_2 } - D_{P_6 P_1 } = b\,f_{P_5 P_2 \_P_5 P_1 } \left(
{\frac{a}{b}} \right),
\]

\noindent where
\begin{align}
f_{P_5 P_2 \_P_5 P_1 } (x) & = \frac{13}{9}f_{P_5 P_2 } (x) - f_{P_5 P_1 }
(x) \notag\\
& = \frac{\left( {\sqrt x - 1} \right)^4\left( {\begin{array}{l}
 4x^5 + 16x^{9 / 2} + 44x^4 + 88x^{7 / 2} + 139x^3 + \\
 + 162x^{5 / 2} + 139x^2 + 88x^{3 / 2} + 16\sqrt x + 44x + 4 \\
 \end{array}} \right)}{9\left( {x^2 + 1} \right)\left( {x^3 + 1}
\right)\left( {\sqrt x + 1} \right)^2}. \notag
\end{align}

Since $f_{P_5 P_2 \_P_5 P_1 } (x) > 0$, $\forall x > 0,\,x \ne 1$, hence
proving the required result.

\item \textbf{For $\bf{D_{P_6 P_2 } \leqslant \frac{12}{7}D_{P_5 P_3 }}:$} By considering the
function $g_{P_6 P_2 \_P_5 P_3 } (x) = {{f}''_{P_6 P_2 } (x)}
\mathord{\left/ {\vphantom {{{f}''_{P_6 P_2 } (x)} {{f}''_{P_5 P_3 } (x)}}}
\right. \kern-\nulldelimiterspace} {{f}''_{P_5 P_3 } (x)}$, we get $\beta
_{P_6 P_2 \_P_5 P_3 } = g_{P_6 P_2 \_P_5 P_3 } (1) = \frac{12}{7}$, where
${f}''_{P_6 P_2 } (x)$ and ${f}''_{P_5 P_3 } (x)$ are as given by (\ref{eq22}) and
(\ref{eq30}) respectively. We can write
\[
\frac{12}{7}D_{P_5 P_3 } - D_{P_6 P_2 } = b\,f_{P_5 P_3 \_P_6 P_2 } \left(
{\frac{a}{b}} \right),
\]

\noindent where
\begin{align}
f_{P_5 P_3 \_P_6 P_2 } (x) & = \frac{16}{9}f_{P_5 P_3 } (x) - f_{P_6 P_2 }
(x) \notag\\
& = \frac{\left( {\sqrt x - 1} \right)^4\left( {\begin{array}{l}
 5x^4 + x^{7 / 2} + 17x^3 + 22x^{5 / 2} + \\
 + 38x^2 + 22x^{3 / 2} + 17x + \sqrt x + 5 \\
 \end{array}} \right)}{7\left( {x^2 + 1} \right)\left( {x^3 + 1}
\right)\left( {\sqrt x + 1} \right)^2}.\notag
\end{align}

Since $f_{P_5 P_3 \_P_6 P_2 } (x) > 0$, $\forall x > 0,\,x \ne 1$, hence
proving the required result.

\item \textbf{For $\bf{D_{P_5 P_2 } \leqslant \frac{9}{7}D_{P_5 P_3 }}:$} By considering the
function $g_{P_5 P_2 \_P_5 P_3 } (x) = {{f}''_{P_5 P_2 } (x)}
\mathord{\left/ {\vphantom {{{f}''_{P_5 P_2 } (x)} {{f}''_{P_5 P_3 } (x)}}}
\right. \kern-\nulldelimiterspace} {{f}''_{P_5 P_3 } (x)}$, we get $\beta
_{P_5 P_2 \_P_5 P_3 } = g_{P_5 P_2 \_P_5 P_3 } (1) = \frac{9}{7}$, where
${f}''_{P_5 P_2 } (x)$ and ${f}''_{P_5 P_3 } (x)$ are as given by (\ref{eq31}) and
(\ref{eq30}) respectively. We can write
\[
\frac{9}{7}D_{P_5 P_3 } - D_{P_5 P_2 } = b\,f_{P_5 P_3 \_P_5 P_2 } \left(
{\frac{a}{b}} \right),
\]

\noindent where
\begin{align}
f_{P_5 P_3 \_P_5 P_2 } (x) & = \frac{9}{7}f_{P_5 P_3 } (x) - f_{P_5 P_2 } (x) \notag\\
& = \frac{\left( {\sqrt x - 1} \right)^4\left( {\begin{array}{l}
 2x^3 + 6x^{5 / 2} + 16x^2 + \\
 + 21x^{3 / 2} + 16x + 6\sqrt x + 2 \\
 \end{array}} \right)}{7\left( {x^2 + 1} \right)\left( {x - \sqrt x + 1}
\right)\left( {\sqrt x + 1} \right)^2}. \notag
\end{align}

Since $f_{P_5 P_3 \_P_5 P_2 } (x) > 0$, $\forall x > 0,\,x \ne 1$, hence
proving the required result.

\item \textbf{For $\bf{D_{P_5 P_3 } \leqslant \frac{14}{9}D_{P_6 N_2 }}:$} By considering the
function $g_{P_5 P_3 \_P_6 N_2 } (x) = {{f}''_{P_5 P_3 } (x)}
\mathord{\left/ {\vphantom {{{f}''_{P_5 P_3 } (x)} {{f}''_{P_6 N_2 } (x)}}}
\right. \kern-\nulldelimiterspace} {{f}''_{P_6 N_2 } (x)}$, we get $\beta
_{P_5 P_3 \_P_6 N_2 } = g_{P_5 P_3 \_P_6 N_2 } (1) = \frac{14}{9}$, where
${f}''_{P_5 P_3 } (x)$ and ${f}''_{P_6 N_2 } (x)$ are as given by (\ref{eq30}) and
(\ref{eq17}) respectively. We can write
\[
\frac{9}{7}D_{P_6 N_2 } - D_{P_5 P_3 } = b\,f_{P_6 N_2 \_P_5 P_3 } \left(
{\frac{a}{b}} \right),
\]

\noindent where
\begin{align}
f_{P_5 P_3 \_P_6 N_2 } (x) & = \frac{14}{9}f_{P_6 N_2} (x) - f_{P_5 P_3 }
(x) \notag\\
& = \frac{\left( {\sqrt x - 1} \right)^4\left( {\begin{array}{l}
 x^6 + 268x^{11 / 2} + 1143x^5 + 782x^{9 / 2} + \\
 + 2597x^4 + 2014x^{7 / 2} + 4478x^3 + 2014x^{5 / 2} + \\
 + 2597x^2 + 782x^{3 / 2} + 1143x + 268\sqrt x + 1 \\
 \end{array}} \right)}{18\left( {x^3 + 1} \right)\left( {\sqrt x + 1}
\right)^2}. \notag
\end{align}

Since $f_{P_6 N_2 \_P_5 P_3 } (x) > 0$, $\forall x > 0,\,x \ne 1$, hence
proving the required result.

\item \textbf{For $\bf{D_{P_6 N_2 } \leqslant \frac{9}{4}D_{P_6 S}} :$} By considering the
function $g_{P_6 N_2 \_P_6 S} (x) = {{f}''_{P_6 N_2 } (x)} \mathord{\left/
{\vphantom {{{f}''_{P_6 N_2 } (x)} {{f}''_{P_6 S} (x)}}} \right.
\kern-\nulldelimiterspace} {{f}''_{P_6 S} (x)}$, we get $\beta _{P_6 N_2
\_P_6 S} = g_{P_6 N_2 \_P_6 S} (1) = \frac{9}{4}$, where ${f}''_{P_6 N_2 }
(x)$ and ${f}''_{P_6 S} (x)$ are as given by (\ref{eq17}) and (\ref{eq16}) respectively.
We can write
\[
\frac{9}{4}D_{P_6 S} - D_{P_6 N_2 } = b\,f_{P_6 S\_P_6 N_2 } \left(
{\frac{a}{b}} \right),
\]

\noindent where
\[
f_{P_6 S\_P_6 N_2 } (x) = \frac{9}{4}f_{P_6 S} (x) - f_{P_6 N_2 } (x)
 = \frac{k_{40} (x)}{8\left( {x + 1} \right)},
\]

\noindent with
\[
k_{40} (x) = 2\sqrt {2x + 2} \left( {\sqrt x + 1} \right)\left( {x + 1}
\right) + 10x^2 + 10 - 9\left( {x + 1} \right)\sqrt {2x^2 + 2} .
\]

In order to prove $k_{40} (x) > 0$, $\forall x > 0,\,x \ne 1$, we shall
apply twice the argument given in part 2. Let us consider
\begin{align}
h_{40} (x) & = \left[ {2\sqrt {2x + 2} \left( {\sqrt x + 1} \right)\left( {x +
1} \right) + 10x^2 + 10} \right]^2 - \left[ {9\left( {x + 1} \right)\sqrt
{2x^2 + 2} } \right]^2 \notag\\
&  = 40\sqrt {2x + 2} \left( {\sqrt x + 1} \right)\left( {x^2 + 1}
\right)\left( {x + 1} \right) \notag\\
& \hspace{20pt} - \left[ {46x^4 + 8\left( {x + 1} \right)^3\left( {\sqrt x - 1} \right)^2 +
260x^3 + 28x^2 + 260x + 46} \right]. \notag
\end{align}

Applying again over $h_{40} (x)$ the argument given in part 2, we get
\begin{align}
h_{40a} (x) & = \left[ {40\sqrt {2x + 2} \left( {\sqrt x + 1} \right)\left(
{x^2 + 1} \right)\left( {x + 1} \right)} \right]^2 \notag\\
& \hspace{20pt}  - \left[ {46x^4 + 8\left( {x + 1} \right)^3\left( {\sqrt x - 1} \right)^2 +
260x^3 + 28x^2 + 260x + 46} \right]^2.\notag\\
& = 4\left( {\sqrt x - 1} \right)^4\left( {\begin{array}{l}
 71 + 2316\sqrt x + 11960x^{5 / 2} + 4090x + 11960x^{7 / 2} + \\
 + 11180x^{3 / 2} + 2316x^{11 / 2} + 11180x^{9 / 2} + 12021x^4 + \\
 + 12021x^2 + 6004x^3 + 4090x^5 + 71x^6 \\
 \end{array}} \right).\notag
\end{align}

Since $h_{40a} (x) > 0$ giving $h_{40} (x) > 0$, $\forall x > 0,\,x \ne 1$
and consequently, we have $k_{40} (x) > 0$, $\forall x > 0,\,x \ne 1$. Thus
we have $f_{P_6 S\_P_6 N_2 } (x) > 0$, $\forall x > 0,\,x \ne 1$, thereby
proving the required result.

\item \textbf{For $\bf{D_{P_6 S} \leqslant D_{AG} }:$} By considering the function $g_{P_6
S\_AG} (x) = {{f}''_{P_6 S} (x)} \mathord{\left/ {\vphantom {{{f}''_{P_6 S}
(x)} {{f}''_{AG} (x)}}} \right. \kern-\nulldelimiterspace} {{f}''_{AG}
(x)}$, we get $\beta _{P_6 S\_AG} = g_{P_6 S\_AG} (1) = 1$, where
${f}''_{P_6 S} (x)$ and ${f}''_{AG} (x)$ are as given by (\ref{eq16}) and (\ref{eq17})
respectively. We can write
\[
D_{AG} - D_{P_6 S} = b\,f_{AG\_P_6 S} \left( {\frac{a}{b}} \right),
\]

\noindent where
\[
f_{AG\_P_6 S} (x) = f_{AG} (x) - f_{P_6 S} (x)
 = \frac{k_{41} (x)}{2\left( {x + 1} \right)},
\]

\noindent with
\[
k_{41} (x) = \sqrt {2x^2 + 2} \left( {x + 1} \right) - \left[ {\left( {x -
1} \right)^2 + 2\left( {x + 1} \right)\sqrt x } \right].
\]

Now we shall show that $k_{41} (x) > 0$, $\forall x > 0,\,x \ne 1$. Let us consider
\[
h_{41} (x) = \left[ {\sqrt {2x^2 + 2} \left( {x + 1} \right)} \right]^2 -
\left[ {\left( {x - 1} \right)^2 + 2\left( {x + 1} \right)\sqrt x }
\right]^2
\]

After simplifications we have
\[
h_{41} (x) = \left( {\sqrt x + 1} \right)^2\left( {\sqrt x - 1} \right)^6.
\]

Since $h_{41} (x) > 0$, giving $k_{41} (x) > 0$, $\forall x > 0,\,x \ne 1$.
Thus we have $f_{AG\_P_6 S} (x) > 0$, $\forall x > 0,\,x \ne 1$, thereby
proving the required result.

Parts 34-41 prove the sequences of inequalities appearing in (\ref{eq45}).
\end{enumerate}
\end{proof}

\begin{remark} (i) In view of Theorem 3.1, we have the following improvement over the
inequalities (\ref{eq11}):
\[
D_{SA} \leqslant \left\{ {\begin{array}{l}
 \tfrac{1}{3}D_{SH} \leqslant \tfrac{1}{2}D_{AH} \leqslant 4D_{N_2 N_1 }
\leqslant \tfrac{4}{3}D_{N_2 G} \leqslant D_{AG} \leqslant 4D_{AN_2 } \\\\
 \left\{ {\begin{array}{l}
 \tfrac{4}{5}D_{SN_2 } \\
 \tfrac{3}{4}D_{SN_3 } \\
 \end{array}} \right\} \leqslant \tfrac{2}{3}D_{SN_1 } \leqslant
\tfrac{1}{2}D_{SG} \\
 \end{array}} \right..
\]

\noindent (ii) The results appearing Parts 1-41 bring us some very interesting measures given by
\[
V_k (a,b) = b\,f_k \left( {a \mathord{\left/ {\vphantom {a b}} \right.
\kern-\nulldelimiterspace} b} \right),
\,
k = 1,2,3,4,
\]

\noindent where
\begin{align}
f_1 (x) & = \frac{\left( {x + 1} \right)^2\left( {x - 1} \right)^4}{\left(
{x^3 + 1} \right)\left( {x^2 + 1} \right)},\notag\\
f_2 (x) & = \frac{\left( {\sqrt x - 1} \right)^2}{x + 1}, \notag\\
f_3 (x) & = \frac{\sqrt x \left( {\sqrt x - 1} \right)^4}{\left( {x + 1}
\right)\left( {\sqrt x + 1} \right)^2} \notag\\
\intertext{and}
f_4 (x) & = \frac{\left( {\sqrt x - 1} \right)^4\left[ {\left( {\sqrt x - 1}
\right)^2 + \sqrt x } \right]}{12\left( {x + 1} \right)}.\notag
\end{align}

The measures $V_1 (a,b)$ is due to part 1, the measure $V_2 (a,b)$ is due to
parts 6, 8, 13, 23 and 24, the measure $V_3 (a,b)$ is due to parts 9 and 16
and the measure $V_4 (a,b)$ is due to parts 14 and 15.
\end{remark}

\section{Connections with Information Measures}

\[
\Gamma _n = \left\{ {P = (p_1 ,p_2 ,...,p_n )\left| {p_i > 0,\sum\limits_{i
= 1}^n {p_i = 1} } \right.} \right\},
\,
n \geqslant 2,
\]

\noindent
be the set of all complete finite discrete probability distributions. For
all $P,Q \in \Gamma _n $, the following inequalities are already proved by
author \cite{tan2}:
\begin{equation}
\label{eq46}
\frac{1}{2}D_{AH} \leqslant I \leqslant 4D_{N_2 N_1 } \leqslant
\frac{4}{3}D_{N_2 G}
 \leqslant D_{AG} \leqslant 4\,D_{AN_2 } \leqslant \frac{1}{8}J \leqslant T
\leqslant \frac{1}{16}\Psi,
\end{equation}

\noindent where

\begin{align}
\label{eq47}
I(P\vert \vert Q) & = \frac{1}{2}\left[ {\sum\limits_{i = 1}^n {p_i \ln \left(
{\frac{2p_i }{p_i + q_i }} \right) + } \sum\limits_{i = 1}^n {q_i \ln \left(
{\frac{2q_i }{p_i + q_i }} \right)} } \right],\\
\label{eq48}
J(P\vert \vert Q) & = \sum\limits_{i = 1}^n {(p_i - q_i )\ln \left( {\frac{p_i
}{q_i }} \right)}, \\
\label{eq49}
T(P\vert \vert Q) & = \sum\limits_{i = 1}^n {\left( {\frac{p_i + q_i }{2}}
\right)\ln \left( {\frac{p_i + q_i }{2\sqrt {p_i q_i } }} \right)}\\
\intertext{and}
\label{eq50}
\Psi (P\vert \vert Q) & = \sum\limits_{i = 1}^n {\frac{(p_i - q_i )^2(p_i +
q_i )}{p_i q_i }}.
\end{align}

The measures $I(P\vert \vert Q)$, $J(P\vert \vert Q)$ and $T(P\vert \vert Q)$ are the respectively, the well-know \textit{Jensen-Shannon divergence}, $J - $\textit{divergence }and \textit{arithmetic and geometric mean divergence}. The measure $\Psi (P\vert \vert Q)$ is \textit{symmetric chi-square divergence}. For detailed study on these measures refer to Taneja \cite{tan1, tan2, tan4}. Moreover, $D_{AH} (P\vert \vert Q) = \tfrac{1}{2}\Delta (P\vert \vert Q)$ and $D_{AG} (P\vert \vert Q) = h(P\vert \vert Q)$, where $\Delta (P\vert \vert Q)$ and $h(P\vert \vert Q)$ are the well-known \textit{triangular's} and
\textit{Hellinger's} \textit{discriminations} respectively.

\bigskip
In the theorem below we shall make connections of the classical divergence measures given in (\ref{eq47})-(\ref{eq50}) with the inequalities given in (\ref{eq43}). Moreover the theorem below unifies the inequalities (\ref{eq43}) and (\ref{eq46}).

\begin{theorem} The following inequalities hold:
\begin{align}
& \left\{ {\begin{array}{l}
 \tfrac{2}{5}D_{P_5 H} \\
 \tfrac{2}{3}D_{AP_4 } \leqslant I \\
 \end{array}} \right\} \leqslant 4D_{N_2 N_1 } \leqslant \frac{4}{3}D_{N_2
G} \leqslant D_{AG} \leqslant 4D_{AN_2 } \leqslant \notag\\
& \hspace{20pt} \leqslant \frac{2}{3}D_{P_5 G} \leqslant \left\{ {\begin{array}{l}
 D_{P_5 N_1 } \leqslant \tfrac{6}{5}D_{P_5 N_3 } \leqslant
\tfrac{4}{3}D_{P_5 N_2 } \leqslant 2D_{P_5 A} \\
 \tfrac{1}{8}J \\
 \end{array}} \right\} \leqslant T \leqslant \frac{1}{16}\Psi. \notag
\end{align}
\end{theorem}

\begin{proof} In view of the inequalities given in (\ref{eq43}) and (\ref{eq46}), it sufficient to show the following three inequalities:
\begin{itemize}
\item [(i)] $ D_{AP_4 } \leqslant \frac{3}{2}I;$\
\item [(ii)] $ D_{P_5 G} \leqslant \frac{3}{16}J;$\
\item [(iii)] $ D_{P_5 A} \leqslant \frac{1}{2}T.$
\end{itemize}

Since in each part of the above expressions we have logarithmic form, we shall apply a different approach to show the above three inequalities.

\bigskip
\noindent \textbf{(i) For }$\bf{D_{AP_4 } \leqslant \frac{3}{2}I}$\textbf{:} Let us consider
\[
g_{AP_4 \_I} (x) = \frac{{f}''_{AP_4 } (x)}{{f}''_I (x)} = \frac{12x(x +
1)}{\sqrt x \left( {\sqrt x + 1} \right)^4},
\,
x \ne 1,\;x > 0.
\]

Calculating the first order derivative of the function $g_{AP_4 \_I} (x)$ with respect to $x$, one gets
\begin{equation}
\label{eq51}
{g}'_{AP_4 \_I} (x) = - \frac{6\left( {x^{3 / 2} - 3x + 3\sqrt x - 1}
\right)}{\sqrt x \left( {\sqrt x + 1} \right)^5}
 = - \frac{6\left( {\sqrt x - 1} \right)^3}{\sqrt x \left( {\sqrt x + 1}
\right)^5}\begin{cases}
 { > 0,} & {x < 1} \\
 { < 0,} & {x > 1} \\
\end{cases}
\end{equation}

In view of (\ref{eq51}), we conclude that the function $g_{AP_4 \_I} (x)$ is increasing in $x \in (0,1)$ and decreasing in $x \in (1,\infty )$, and hence
\begin{equation}
\label{eq52}
\beta_{AP_4 \_I} = \mathop {\sup }\limits_{x \in (0,\infty )} g_{AP_4 \_I} (x) =
g_{AP_4 \_I} (1) = \frac{3}{2}.
\end{equation}

By the application of Lemma 3.1 with (\ref{eq52}) we get the required result.

\bigskip
\noindent \textbf{(ii) For }$\bf{D_{P_5 G} \leqslant \frac{3}{16}J}$\textbf{:} Let us
consider
\[
g_{P_5 G\_J} (x) = \frac{{f}''_{P_5 G} (x)}{{f}''_J (x)} = \frac{3\sqrt x
\left( {4x^{3 / 2} + 4\sqrt x + x^2 - 2x + 1} \right)}{4\left( {x + 1}
\right)\left( {\sqrt x + 1} \right)^4}.
\]

Calculating the first order derivative of the function $g_{P_5 G\_J} (x)$ with respect to $x$, one gets
\begin{equation}
\label{eq53}
{g}'_{P_5 G\_J} (x) = - \frac{3\left( {\sqrt x - 1} \right)^3\left( {x^2 + 8x^{3 / 2} + 6x +
8\sqrt x + 1} \right)}{8\sqrt x \left( {\sqrt x + 1} \right)^5\left( {x + 1}
\right)^2}\begin{cases}
 { > 0,} & {x < 1} \\
 { < 0,} & {x > 1} \\
\end{cases}
\end{equation}

In view of (\ref{eq53}), we conclude that the function $g_{P_5 G\_J} (x)$ is increasing in $x \in (0,1)$ and decreasing in $x \in (1,\infty )$, and hence
\begin{equation}
\label{eq54}
\beta_{P_5 G\_J} = \mathop {\sup }\limits_{x \in (0,\infty )} g_{P_5 G\_J} (x) = g_{P_5
G\_J} (1) = \frac{3}{2}.
\end{equation}

By the application of Lemma 3.1 with (\ref{eq54}) we get the required result.

\bigskip
\noindent \textbf{(iii) For }$\bf{D_{P_5 A} \leqslant \frac{1}{2}T}$\textbf{:} Let us
consider
\[
g_{P_5 A\_T} (x) = \frac{{f}''_{P_5 A} (x)}{{f}''_T (x)} = \frac{2\sqrt x
\left( {4x^{3 / 2} + 4\sqrt x + x^2 - 6x + 1} \right)\left( {x + 1}
\right)}{\left( {x^2 + 1} \right)\left( {\sqrt x + 1} \right)^4}.
\]

Calculating the first order derivative of the function $g_{P_5 A\_T} (x)$ with respect to $x$, one gets
\begin{equation}
\label{eq55}
 {g}'_{P_5 A\_T} (x) = - \frac{\left( {\sqrt x - 1} \right)^3\left( {\begin{array}{l}
 8\sqrt x \left( {x^2 + 1} \right)\left( {\sqrt x - 1} \right)^2 + \\
 + x^4 + 14x^3 + 10x^2 + 14x + 1 \\
 \end{array}} \right)}{\sqrt x \left( {\sqrt x + 1} \right)^5\left( {x^2 +
1} \right)^2}\begin{cases}
 { > 0,} & {x < 1} \\
 { < 0,} & {x > 1} \\
\end{cases}
\end{equation}

In view of (\ref{eq55}), we conclude that the function $g_{P_5 A\_T} (x)$ is increasing in $x \in (0,1)$ and decreasing in $x \in (1,\infty )$, and hence
\begin{equation}
\label{eq56}
\beta_{P_5 A\_T} = \mathop {\sup }\limits_{x \in (0,\infty )} g_{P_5 A\_T} (x) = g_{P_5
A\_T} (1) = \frac{1}{2}.
\end{equation}

By the application of Lemma 3.1 with (\ref{eq56}) we get the required result.

\bigskip
Parts (i)-(iii) completes the proof of the Theorem 4.1.
\end{proof}

\bigskip
Here we have referred Lemma 3.1, but its extension for the probability distributions is already proved in \cite{tan2, tan3}.

\end{document}